\documentclass{mytemplate}

\title{Variance estimation after matching or re-weighting}
\author{
Xiang Meng\thanks{Corresponding author: \texttt{xmeng@g.harvard.edu}. Postdoctoral Fellow at Dana-Farber Cancer Institute. Work done while completing Ph.D. in the Department of Statistics, Harvard University.}
\and
Aaron Smith\thanks{Associate Professor, Department of Mathematics and Statistics, University of Ottawa.}
\and
Luke Miratrix\thanks{Professor, Harvard Graduate School of Education.}
}
\begin{document}
\maketitle

\begin{abstract}
This paper develops a variance estimation framework for matching estimators that enables valid population inference for treatment effects. We provide theoretical analysis of a variance estimator that addresses key limitations in the existing literature. While \cite{abadie2006large} proposed a foundational variance estimator requiring matching for both treatment and control groups, this approach is computationally prohibitive and rarely used in practice. Our method provides a computationally feasible alternative that only requires matching treated units to controls while maintaining theoretical validity for population inference. 

We make three main contributions: First, we establish consistency and asymptotic normality for our variance estimator, proving its validity for average treatment effect on the treated (ATT) estimation in settings with small treated samples. Second, we develop a generalized theoretical framework with novel regularity conditions that significantly expand the class of matching procedures for which valid inference is available, including radius matching, M-nearest neighbor matching, and propensity score matching. Third, we demonstrate that our approach extends naturally to other causal inference estimators such as stable balancing weighting methods. 

Through simulation studies across different data generating processes, we show that our estimator maintains proper coverage rates while the state-of-the-art bootstrap method can exhibit substantial undercoverage (dropping from 95\% to as low as 61\%), particularly in settings with extensive control unit reuse. Our framework provides researchers with both theoretical guarantees and practical tools for conducting valid population inference across a wide range of causal inference applications. An R package implementing our method is available at https://github.com/jche/scmatch2.
\end{abstract}

\section{Introduction}

Matching and weighting estimators are fundamental tools in causal inference for estimating treatment effects from observational data. These methods enable researchers to draw population-level inferences about treatment effects by comparing treated units with similar control units based on observed covariates \citep{rosenbaum1983central, rubin1973matching} or by reweighting observations to achieve covariate balance \citep{hirano2003efficient, imbens2004nonparametric}. Valid population inference—the ability to generalize findings beyond the specific sample to the broader population—is crucial for policy decisions and scientific understanding across diverse fields including economics \citep{dehejia1999causal, heckman1997matching}, epidemiology \citep{stuart2010matching}, and policy evaluation \citep{smith2005does}.

\cite{abadie2006large} established the foundational asymptotic theory for matching estimators, revealing their nonstandard behavior due to the fixed number of matches. Unlike other nonparametric treatment effect estimators \citep{heckman1998matching, hirano2003efficient}, matching estimators are generally not $\sqrt{N}$-consistent. This discovery led to important developments in bias correction \citep{abadie2011bias} and inference methods. However, the variance estimator proposed by \cite{abadie2006large} requires matching for both treatment and control groups, making it computationally prohibitive in practice, particularly for average treatment effect on the treated (ATT) estimation where the treated group is typically small relative to the control pool.

Recognizing the computational limitations of analytical approaches, \cite{abadie2008failure} demonstrated that the standard bootstrap fails to provide valid inference for matching estimators, leading to the development of alternative resampling methods. \cite{otsu2017bootstrap} proposed a wild bootstrap procedure that is theoretically valid under certain asymptotic conditions and has become the current state-of-the-art for inference in matching. However, our empirical investigations reveal that this procedure can produce unreliable inference when there is substantial control unit reuse—a common scenario when the same control units appear repeatedly in the matched sets of multiple treated units. This control unit reuse occurs because when matching algorithms seek the best matches in finite samples, high-quality control units are often among the nearest neighbors for multiple treated units, especially in high-dimensional covariate spaces where the curse of dimensionality makes truly distinct matches rare. This dependency structure becomes particularly problematic when constructing confidence intervals using bootstrap methods, as the resampling scheme fails to capture the true variance of the estimator \citep{abadie2012martingale}.

We address this challenge by developing a computationally efficient variance estimation framework that enables valid population inference while addressing the limitations of existing methods. Our approach provides a practical alternative that only requires matching treated units to controls—making it particularly suitable for ATT estimation—while maintaining theoretical validity for population inference. First, we analyze a variance estimator that demonstrates robust performance in simulations when the wild bootstrap fails to produce valid confidence intervals under substantial control unit reuse. This procedure is based on a Wald-type confidence interval using a variance estimator that has been used in empirical work (e.g., \cite{che2024caliper, keele_hospital_2023}), but whose theoretical properties and comparative advantages have not been systematically established. Second, we develop a theoretical framework that rigorously justifies the estimator's validity by establishing its consistency and asymptotic normality, providing practitioners with a sound statistical foundation for population inference.

Our contributions are as follows:

\begin{itemize}
\item \textbf{Theoretical Analysis of a Computationally Efficient Variance Estimator:} We provide rigorous theoretical analysis of a variance estimator that addresses key computational and practical limitations in the existing literature. Unlike the estimator in \cite{abadie2006large} which requires matching within both treatment groups, our approach only requires matching treated units to controls, making it particularly valuable for ATT estimation with small treated samples. Our analysis demonstrates that this estimator is theoretically justified and outperforms current state-of-the-art alternatives in realistic settings. In particular, our estimator remains robust when control units are reused as matches for multiple treated units. Through carefully designed simulation studies, we show that it maintains proper coverage (95\%) while the wild bootstrap method proposed by \cite{otsu2017bootstrap} can fail dramatically (dropping to 61\% coverage) in common scenarios with extensive control unit reuse. We also establish that the estimator possesses heteroskedasticity-consistent properties, drawing important parallels to the well-known Huber-White \cite{white1980heteroskedasticity} robust standard errors in regression analysis.

\item \textbf{Generalized Theoretical Framework:} We show that our framework applies to a variety of matching contexts including nearest neighbor matching, radius matching, propensity score matching, and synthetic control weighting. We introduce two novel conditions to the matching literature that together create a more practical framework for valid population inference. First, our derivative control condition relaxes traditional requirements about how outcome functions can change. While previous work by \cite{abadie2006large} required the outcome function to change at a constant rate everywhere (Lipschitz continuity), our approach allows the rate of change to vary across the covariate space as long as it is appropriately balanced by the size of matched clusters. Second, our shrinking clusters assumption does not specify how quickly matches must improve with sample size, only that they eventually become arbitrarily close. This flexibility accommodates many matching methods beyond the specific $M$-NN approach in earlier work, significantly expanding the applicability of matching methods while ensuring theoretical guarantees for population inference.

\item \textbf{Extensions to Other Causal Inference Estimators:} We demonstrate that the variance estimation framework extends beyond matching to other causal inference methods, particularly weighting estimators \cite{zubizarreta2015stable, wang_minimal_2019}. Through simulations with challenging datasets, we show that our variance estimation approach maintains proper coverage when applied to stable balancing weights, suggesting potential for creating a unified inference framework for both matching and weighting approaches in causal inference.
\end{itemize}

\section{Problem Setup}

\subsection{Model}
\label{sec:model}
We consider a setting with $n$ observations, each representing a unit in our study population. The sample consists of $n_T$ treated units and $n_C$ control units, with $n = n_T + n_C$.

For each unit $i$, we observe a tuple $\{Z_i, Y_i, \mathbf{X}_i\}$ where:
\begin{itemize}
    \item $Z_i \in \{0,1\}$ denotes its binary treatment status.
    \item $Y_i \in \mathbb{R}$ denotes its observed real-valued outcome.
    \item $\mathbf{X}_i \equiv \{X_{1i}, \dots, X_{ki} \}^T \in \mathbb{R}^k$ denotes its $k$-dimensional real-valued covariate vector.
\end{itemize}

We adopt the potential outcomes framework where each unit has two potential outcomes: $Y_i(1)$ and $Y_i(0)$. Here, $Y_i(1)$ represents the outcome if unit $i$ receives treatment, and $Y_i(0)$ represents the outcome if unit $i$ does not receive treatment. The fundamental problem of causal inference is that we only observe one of these potential outcomes for each unit. Specifically, the observed outcome for unit $i$ is $Y_i \equiv (1-Z_i) Y_i(0) + Z_i Y_i(1)$ under the stable unit treatment value assumption (SUTVA). 

We assume the data generating process follows independent and identically distributed (i.i.d.) sampling of the potential outcome tuples $\{Y_i(0), Y_i(1), Z_i, \mathbf{X}_i\}_{i=1}^n$. For each unit $i$, the generic random variables $(Y(0), Y(1), Z, \mathbf{X})$ represent the population distribution from which the observed data are drawn. Throughout the paper, indexed variables (e.g., $\mathbf{X}_i$) refer to specific observations, while non-indexed variables (e.g., $\mathbf{X}$) refer to the generic random variables representing the population distribution.

To proceed with estimation, we make the following assumptions:

\begin{assumption}[Compact support]
\label{assum:covariate}
The covariate vector $\mathbf{X}$ is a $k$-dimensional random vector with components that are continuous random variables, distributed on $\mathbb{R}^k$ with compact support $\mathbb{X}$. The density of $X$ is bounded and bounded away from zero on its support.
\end{assumption}

The compact support assumption helps ensure that the covariate space is well-behaved, which facilitates consistent estimation and rules out pathological cases where the distribution of covariates becomes too sparse or unbounded.

\begin{assumption}[Unconfoundedness and overlap \citep{rubin1974estimating}] \label{assum:unconfoundedness} 
For almost every $x \in \mathbb{X}$: \begin{enumerate} \item $(Y(1), Y(0)) \indep Z \mid \mathbf{X}$ \item $\Pr(Z=1 \mid X = x) < 1 - \eta$ for some $\eta > 0$ \end{enumerate} 
\end{assumption}

This assumption states that, conditional on the observed covariates, treatment assignment is independent of the potential outcomes, and that both treated and control units are sufficiently represented across the covariate space.

Importantly, the treatment indicators $Z_i$ are randomly drawn according to the treatment assignment mechanism, which implies that the number of treated units $n_T$ is a random quantity even when the total sample size $n$ is fixed. This leads to our next assumption:

\begin{assumption}[Sampling]
\label{assum:sampling}
Conditional on $Z_i=z$, the sample consists of independent draws from the distribution of $(Y, X | Z=z)$ for $z \in \{0,1\}$. As the sample size $n \to \infty$, we have $n_T^r/n_C \to \theta$ for some $r \geq 1$ and $0 < \theta < \infty$.
\end{assumption}

We further assume a model where potential outcomes are generated as:
\begin{align*}
	Y_i(0) &= f_0(\mathbf{X}_i) + \epsilon_{0,i} \\ 
	Y_i(1) &= f_1(\mathbf{X}_i) + \epsilon_{1,i}
\end{align*}

Here, $f_0(\mathbf{X}) \equiv E[Y(0) | \mathbf{X}]$ and $f_1(\mathbf{X}) \equiv E[Y(1) | \mathbf{X}]$ are the true conditional expectation functions of the potential outcomes under control and treatment, respectively (often referred to as ``response surfaces" in the causal inference literature \citep{hahn2020bayesian, hill2011bayesian}). The error terms $\epsilon_{0,i}$ and $\epsilon_{1,i}$ represent the deviations of the individual potential outcomes from their respective conditional expectations, with conditional variances $\sigma_{0,i}^2$ and $\sigma_{1,i}^2$ respectively. Further distributional assumptions about these error terms are detailed in Section~\ref{sec:err-var-assumptions}.

\subsection{Estimand and the Estimator}
\label{sec:estimand}
We define $\tau(\mathbf{X}_i) = f_1(\mathbf{X}_i) - f_0(\mathbf{X}_i)$ as the systematic treatment effect, which represents the systematic component of the treatment effect for units with covariates $\mathbf{X}_i$. Note that the individual treatment effect $Y_i(1) - Y_i(0) = \tau(\mathbf{X}_i) + (\epsilon_{1,i} - \epsilon_{0,i})$ includes both this systematic component and an idiosyncratic component.
We consider the estimand to be the population average treatment effect on the treated (ATT)
\[
    \tau \;=\; E\bigl[f_1(X_i) - f_0(X_i)\mid Z_i = 1\bigr],
\]
We write the set of all treated units' indices as $\mathcal{T} = \{i: Z_i=1\}$, the set of all control units' indices as $\mathcal{C} = \{i: Z_i=0\}$, and $t \in \mathcal{T}$, $j \in  \mathcal{C}$ as individual treated and control units respectively. We denote the set of indices of control units matched to a treated unit $t \in \mathcal{T}$ as $\mathcal{C}_t = \{j \in \mathcal{C}: \text{ unit } j \text{ is matched to unit } t\}$, which is determined by a matching procedure that maps the observed data to these match sets.
Finally, we denote the size of a set $\mathcal{S}$ as $|\mathcal{S}|$.

A matching procedure pairs each treated unit with one or more control units that have similar covariate values, thus approximating the counterfactual outcome for the treated unit. Given such a matching procedure, we define the matching estimators for the ATT:
\begin{equation}
\label{eq:matching-estimator}
    \hat{\tau}(w) = \frac{1}{n_T} \sum_{t \in \mathcal{T}} \big(Y_t - \sum_{j \in \mathcal{C}_t} w_{jt} Y_j \big).
\end{equation}
where $w_{jt} \in [0,1]$ is the weight assigned to the matched control unit $j$ for treated unit $t$, with $\sum_{j \in \mathcal{C}_t} w_{jt} = 1$ for each $t \in \mathcal{T}$. For example, in $M$-nearest neighbor ($M$-NN) matching \citep{rubin1973matching, abadie2006large, stuart2010matching}, $\mathcal{C}_t$ consists of the closest $M$ neighbors to unit $t$ based on covariate distance, and each neighbor receives equal weight $w_{jt} = 1/M$. Another example is the synthetic control approach in \cite{che2024caliper}, which first obtains $\mathcal{C}_t$ local radius matching, and then determines the weights $w_{jt}$ by solving a convex optimization problem that minimizes the distance between the treated unit's covariates $\mathbf{X}_t$ and the weighted average of control units' covariates $\sum_{j \in \mathcal{C}_t} w_{jt} \mathbf{X}_j$.

Define the matching radius for a treated unit $t$ with covariate value $\mathbf{X}_t$ as:
\[
r\left(\mathcal{C}_t\right)=\sup _{j \in \mathcal{C}_t}\left\|\mathbf{X}_t - \mathbf{X}_j\right\|.
\]

This radius represents the maximum distance between a treated unit and any of its matched controls. The probabilistic properties of this radius will be crucial for establishing our theoretical results.

\begin{assumption}[Exponential Tail Condition]
\label{assum:exponential-tail}
The matching radius satisfies:
\[
P\left(n_C^{1/k}r\left(\mathcal{C}_t\right) > u\right) \leq C_1  \exp\left(-C_2 u^k\right),
\]
where $C_1, C_2$ are positive constants, $k$ is the dimension of the covariate space, and $M$ is the number of matches.
\end{assumption}

This assumption requires that the probability of having a large (scaled) matching radius decays exponentially, which ensures that the matches become increasingly accurate as the sample size grows. This is a more precise characterization than simply requiring that clusters shrink asymptotically, as it specifies the rate at which the tail probability diminishes.

Several common matching methods satisfy the exponential tail condition under appropriate implementation, including $M$-NN matching and radius matching. When using fixed $M$-NN matching, the exponential tail condition is satisfied provided that the density of covariates is bounded and has overlapped support, as established by \cite{abadie2006large}. In this approach, each treated unit is matched to its $M$ closest control units based on covariate distance, with each neighbor receiving equal weight $w_{jt} = 1/M$. 

Another approach is radius matching, where we set the matching radius as $D(n_C) = c n_C^{-1/k}$. This choice ensures two important properties. First, the probability of obtaining at least one match for a treated unit approaches 1 as $n_C \rightarrow \infty$. To see why, note that the expected number of controls in a ball of radius $D(n_C)$ around a treated unit with covariates $\mathbf{X}_t$ is approximately
\[
n_C \cdot f(\mathbf{X}_t)\, v_k\,[c\,n_C^{-1/k}]^k \approx f(\mathbf{X}_t)\,v_k\,c^k,
\]
where $v_k$ is the volume of the unit ball in $\mathbb{R}^k$ and $f(\mathbf{X}_t)$ is the density at $\mathbf{X}_t$. This expression remains bounded away from zero when $c$ is chosen appropriately and $f(\mathbf{X}_t) > 0$, ensuring that matches are found with high probability.

Second, this matching scheme has an exponential tail for the scaled discrepancy $n_C^{1/k}\|\mathbf{X}_j - \mathbf{X}_t\|$ between any treated unit $t$ and any matched control unit $j \in \mathcal{C}_t$. By analyzing the order statistics of nearest neighbor distances and applying large-deviation bounds, we can show that
\[
\Pr\bigl(n_C^{1/k}\|\mathbf{X}_j - \mathbf{X}_t\| > u\bigr) \le C_1 \exp(-C_2 u^k)
\]
for some constants $C_1, C_2>0$, thus satisfying the exponential tail condition.

In the next section, we address the inference problem, focusing on the asymptotic normality of the matching estimator and the decomposition of its variance components. This provides the foundation for constructing valid confidence intervals. Following this, we turn our attention to the crucial challenge of variance estimation. We introduce a consistent estimator that accounts for both homoskedastic and heteroskedastic error structures, refining previous approaches to improve efficiency and robustness.

\section{The Inference Problem}
\label{sec:inference-problem}
To construct valid confidence intervals for our matching estimator $\hat{\tau}$, we require asymptotic normality of the form:
\[
\frac{\sqrt{n_T}\left( \hat{\tau} - \tau \right)}{V^{-1/2}}
\xrightarrow{d}
N(0,1).
\]
The difference between the matching estimator $\hat{\tau}$ and the estimand $\tau$ can be decomposed into three components:
\begin{align}
\hat{\tau} - \tau = \hat{\tau} - \tau_{\text{SATT}} + \tau_{\text{SATT}} - \tau = B_n + E_n + P_n  \label{eq:hat-tau-breakdown}
\end{align}
where $\tau_{\text{SATT}}$ is the sample average treatment effect on the treated (SATT):

\begin{align*}
\tau_{\text{SATT}} &= \frac{1}{n_T} \sum_{t \in \mathcal{T}} \bigl(f_1(X_t) - f_0(X_t)\bigr).\\
B_n &= \frac{1}{n_T} \sum_{t \in \mathcal{T}} \sum_{j \in \mathcal{C}_t} 
w_{jt} \bigl( f_0(X_t) - f_0(X_j) \bigr) \\
&\text{represents bias from imperfect covariate matching.} \\[0.5em]
E_n &= \frac{1}{n_T} \sum_{t \in \mathcal{T}} \Bigl(\epsilon_t - \sum_{j \in \mathcal{C}_t} w_{jt} \epsilon_j\Bigr) \\
&= \frac{1}{n_T} \sum_{t \in \mathcal{T}} \epsilon_t - \frac{1}{n_T} \sum_{j \in \mathcal{C}} w_{j} \epsilon_j \\
&\text{captures measurement error from random variation in unobserved factors.} \\[0.5em]
P_n &= \tau_{\text{SATT}} - \tau \\
&\text{measures representation error between sample and population treatment effects.}
\end{align*}

We now analyze the key components of our inference framework in detail. Section 3.1 examines the bias term and its asymptotic behavior under different matching schemes. Section 3.2 introduces critical assumptions about error variance that underpin our theoretical results. Section 3.3 develops the variance decomposition, separating contributions from measurement error and population heterogeneity. Together, these elements establish the foundation for our central limit theorem.

\subsection{The Bias Term and Its Convergence Rate}
\label{sec:the-bias-term}
A crucial challenge in establishing the asymptotic normality of matching estimators is the slow convergence rate of the bias term. Following \cite{abadie2006large}, under regularity conditions on data distribution introduced at Section~\ref{sec:model}, this bias term converges at a rate of $O_p(n_T^{-1/k})$, where $k$ is the dimension of the covariate space. This rate is typically slower than the $n_T^{-1/2}$ rate required for standard asymptotic normality results.

\begin{proposition}[Bias convergence rate]
Under Assumptions \ref{assum:covariate}, \ref{assum:unconfoundedness}, and \ref{assum:sampling}, if $f_0(x)$ is Lipschitz continuous on $\mathbb{X}$, then
$$Bias = O_p(n_T^{-1/k})$$
\end{proposition}

This slow convergence rate of the bias term necessitates explicit bias correction for valid inference. While various approaches to bias correction exist, including the method proposed by \cite{abadie2011bias}, our focus in this paper will be on variance estimation conditional on a bias correction procedure.

\subsection{Error Variance Assumptions}
\label{sec:err-var-assumptions}

To analyze the large-sample behavior of our variance estimator, we place structure on the conditional variance of the potential outcomes. This section introduces the regularity conditions we require on the variance functions $\sigma_0^2(x)$ and $\sigma_1^2(x)$.

Let us denote the conditional variances of the potential outcomes as:
\begin{equation}
\label{eq:equal-var}
\begin{aligned}
\sigma_{0,i}^2 
&\;=\; 
E\bigl[\bigl(Y_i(0) \;-\; f_0(\mathbf{X}_i)\bigr)^2 \,\big|\; \mathbf{X}_i\bigr] = E[\epsilon_{0,i}^2 \,|\, \mathbf{X}_i], \\
\sigma_{1,i}^2 
&\;=\; 
E\bigl[\bigl(Y_i(1) \;-\; f_1(\mathbf{X}_i)\bigr)^2 \,\big|\; \mathbf{X}_i\bigr] = E[\epsilon_{1,i}^2 \,|\, \mathbf{X}_i].
\end{aligned}
\end{equation}

We now define a class of variance functions with properties that enable consistent estimation in the matched setting.

\begin{definition}[Regular variance function]
\label{def:regular-variance}
A function $\sigma^2 : \mathcal{X} \to \mathbb{R}_{+}$ is said to be a \textit{regular variance function} if it satisfies the following:

\begin{itemize}
    \item \textbf{Uniform continuity.} $\sigma^2(\cdot)$ is uniformly continuous (or Lipschitz) on the support $\mathcal{X} \subset \mathbb{R}^d$ of $X$.

    \item \textbf{Boundedness.} There exist constants $0 < \sigma^2_{\min} < \sigma^2_{\max} < \infty$ such that
    \[
    \sigma^2_{\min} \leq \sigma^2(x) \leq \sigma^2_{\max} \quad \text{for all } x \in \mathcal{X}.
    \]

    \item \textbf{Higher-order moment bound.} There exists a constant $C < \infty$ and an exponent $\delta > 0$ such that 
    \[
    \sup_{x \in \mathcal{X}}
    \mathbb{E}\bigl[\bigl|\epsilon_i\bigr|^{\,2+\delta} 
      \;\big|\; X_i = x
    \bigr] 
    \;\le\; C.
    \]
\end{itemize}
\end{definition}

The first condition ensures that matched units have similar variances. Specifically, for any matching scheme with \(\| X_{tj} - X_{t}\|\to 0\) (as guaranteed by Assumption~\ref{assum:exponential-tail}), we have \(\sigma^2(X_{tj}) \to \sigma^2(X_{t})\). Hence, \(\sigma_j^2 \approx \sigma_t^2\) for \(j \in \mathcal{C}_t\) whenever \(\mathcal{C}_t\) is constructed by matching on \(X\). In particular,
\[
\max_{j \in \mathcal{C}_t}
\bigl|\sigma^2(X_{tj}) \;-\; \sigma^2(X_{t})\bigr|
\;\;\xrightarrow[]{} 0,
\]
provided that 
\(\max_{j \in \mathcal{C}_t} \|X_{tj} - X_{t}\| \to 0\).  
Definition~\ref{def:regular-variance} generalizes Assumption 4.1 in \cite{abadie2006large}, which assumes Lipschitz continuity.

The Boundedness condition ensures that the conditional variance is bounded away from both zero and infinity. The lower bound prevents degeneracy in the asymptotic distribution and ensures that confidence intervals have positive width. While it is theoretically possible for the variance to approach zero, this would imply that all outcome variability is explained by covariates and there is no residual noise — i.e., $\sigma^2(x) \to 0$ for all $x$. In this case, the contribution of the error term to the sampling variability of $\hat{\tau}$ would vanish. The resulting inference problem becomes degenerate: estimation is still possible, but all uncertainty would stem entirely from treatment effect heterogeneity, not from residual noise. This setting is simpler but often unrealistic, as in practice we typically expect some irreducible measurement error in outcomes. The upper bound limits the influence of outliers, which is needed to establish convergence rates.

The third condition imposes a uniform bound on a higher-order conditional moment of the errors. This assumption is standard in high-dimensional estimation and facilitates the use of maximal inequalities and uniform convergence tools.

We now formally state the assumption we make on the conditional variances of the potential outcomes:

\begin{assumption}[Regular error variances]
\label{assum:regular-variance}
We assume that both $\sigma_0^2(x)$ and $\sigma_1^2(x)$ are regular variance functions, as defined above.
\end{assumption}

\medskip

Together, the three properties of regular variance functions ensure that both the level (expected magnitude of the errors) and the variability (how much the errors fluctuate around their means) of the error process are well-behaved across the full range of covariates. This structure plays a key role in enabling consistent variance estimation under matching, as we will see in the following sections.

\subsection{Decomposition of Asymptotic Variance Components}
\label{sec:asymptotic-variance}

We now analyze the asymptotic variance of the matching estimator $\hat{\tau}$. Recall from Equation~\eqref{eq:hat-tau-breakdown} that the estimation error decomposes as: $\hat{\tau} - \tau = B_n + E_n + P_n,$ where $B_n$ captures bias from covariate mismatch \footnote{Note that while $B_n$ is stochastic, we subtract it in our inference as discussed in Section~\ref{sec:the-bias-term}.}, $E_n$ captures sampling error due to residual outcome noise, and $P_n$ captures the discrepancy between the sample and population average treatment effect on the treated (ATT). As discussed earlier, under appropriate regularity conditions and bias correction, the asymptotic distribution of $\hat{\tau}$ is primarily governed by the variability in the two stochastic terms: $E_n$ and $P_n$.

\paragraph{Measurement Error Component \( V_E \).}
We first consider the component due to residual outcome noise. Conditional on the covariates $\bX$ and treatment assignment vector $\mathbf{Z}$, the variance of $E_n$ is given by:
\begin{equation}
\label{eq:CMSE}
\begin{split}
    V_E 
    :=& \mathbb{E}[E_n^2 \mid \bX, \mathbf{Z}] \\ 
    =& \frac{1}{n_T^2} \left(\sum_{t \in \mathcal{T}} \sigma_{1,t}^2 + \sum_{j \in \mathcal{C}} (w_j)^2 \sigma_{0,j}^2 \right),
\end{split}
\end{equation}
where $w_j = \sum_{t \in \mathcal{T}} w_{jt}$ is the total weight assigned to control unit $j$ across all matched treated units.

This decomposition reflects how residual variance enters the estimator: treated units contribute through their own variances, and controls contribute via squared weight accumulation. Reused controls (with large \( w_j \)) disproportionately affect the overall variance. Prior work including \cite{kallus2020generalized} and \cite{che2024caliper} use this variance structure to study the bias-variance tradeoff in matching estimators—highlighting that tighter matches (which reduce bias) can increase variance due to heavy reuse of control units.

\paragraph{Population Heterogeneity Component \( V_P \).}
The second term \( P_n = \tau_{\text{SATT}} - \tau \) captures how the realized sample of treated units may differ from the target population of treated units. That is, even if outcomes were observed without error, the sample ATT may deviate from the population ATT due to treatment effect heterogeneity.

We define the population heterogeneity component as the variance of this discrepancy:
\[
V_P := \text{Var}(P_n) = \frac{1}{n_T} \mathbb{E}\left[\left(\tau(X_i) - \tau\right)^2 \mid Z_i = 1\right].
\]

This expression shows that \( V_P \) depends on the heterogeneity of treatment effects among treated units. When treatment effects are homogeneous (i.e., \( \tau(X_i) = \tau \) for all \( i \)), this variance term vanishes. In general, however, this component plays a non-trivial role in determining the total asymptotic variance of the estimator.

\subsection{The Central Limit Theorem}
\label{sec:CLT}
We now present our main asymptotic normality result, which forms the basis for valid inference.
\begin{theorem}[Central Limit Theorem]
\label{thm:clt}
Under Assumptions \ref{assum:covariate}, \ref{assum:unconfoundedness}, \ref{assum:sampling}, \ref{assum:exponential-tail} and  \ref{assum:regular-variance}, as $n_T \rightarrow \infty$:
\[
    \frac{\sqrt{n_T}\,\bigl( \hat \tau - B_n - \tau \bigr)}{V^{-1/2}} 
    \;\;\xrightarrow{d}\;\;  
    N(0,1),
\]
where 
\[
    V = n_T \cdot (V_E + V_P).
\]
\end{theorem}
In the special case where the dimension of the covariate space satisfies $k \le 2$, the bias term $B_n$ becomes negligible at a faster rate, yielding:
\[
    \frac{\sqrt{n_T}\,\bigl( \hat \tau - \tau \bigr)}{V^{-1/2}} 
    \;\;\xrightarrow{d}\;\;  
    N(0,1).
\]

This theorem generalizes the seminal results of \cite{abadie2006large}, which were limited to M-NN matching with uniform weights ($w_{jt} = 1/M$). Our framework makes two significant advances: (a) it accommodates arbitrary matching procedures including radius matching, caliper matching, and optimal matching; and (b) it allows for flexible weighting schemes such as kernel weights, bias-corrected weights, and synthetic control weights. This increased flexibility enables practitioners to choose matching methods that better balance bias reduction and variance minimization for their specific applications.

For practical implementation of inference procedures, we develop a consistent estimator $\hat{V}$ for the asymptotic variance $V$ in the subsequent section. By Slutsky's theorem, this will yield:
\[
    \frac{\sqrt{n_T}\,\bigl( \hat \tau - B_n - \tau \bigr)}{\hat V^{-1/2}} 
    \;\;\xrightarrow{d}\;\;  
    N(0,1).
\]

This result provides the foundation for constructing asymptotically valid confidence intervals for the treatment effect. The next section will focus on the consistent estimation of the variance components required for implementation.

\section{The Standard Error Estimator}
To establish valid inference for matching estimators, we analyze a standard error estimation strategy that accommodates both homogeneous and heterogeneous error structures. This approach, previously used in establishing local radius matching in \cite{che2024caliper}  lacks thorough theoretical justification. It is different from existing methods by relaxing traditional assumptions while maintaining consistency under general matching procedures. Our theoretical analysis provides the missing foundation for this estimator's widespread application.

The organization of this section is as follows: In Section \ref{sec:derivative-condition}, we introduce the derivative control condition that generalizes previous assumptions; in Section \ref{sec:variance-estimator}, we present the formal variance estimator and establish its consistency; and in Section \ref{sec:comparison}, we compare our approach with previous estimators in the literature, highlighting its practical advantages.

\subsection{Derivative Control Condition}
\label{sec:derivative-condition}

Before introducing the variance estimator, we first establish a key theoretical condition that enables our analysis. The derivative control condition presented below is more general than the Lipschitz continuity assumption on $f$ used in \cite{abadie2006large}, and allows for broader applicability in settings where $f'$ is not uniformly bounded.

\setcounter{assumption}{5}
\begin{assumption}[Derivative control]
\label{assum:derivative-control}
Let $f$ be differentiable on the support of $X$, and denote its derivative by $f': \mathcal{X}\to \mathbb{R}$. 
There exists a constant $M < \infty$ (possibly depending on $n$) such that, for all $t \in \mathcal{T}$,
\[
\sup_{x \in \mathcal{C}_t} \bigl|f'(x)\bigr| \cdot r(\mathcal{C}_t) \le M,
\]
or equivalently,
\[
\sup_{t \in \mathcal{T}} \left[ \sup_{x \in \mathcal{C}_t} \bigl|f'(x)\bigr| \cdot r(\mathcal{C}_t) \right] < \infty.
\]
\end{assumption}

This condition ensures that in regions where the derivative $f'(x)$ is large, the matching clusters $\mathcal{C}_t$ are sufficiently tight—so that the product of local slope magnitude and cluster size remains uniformly bounded. In contrast to the Lipschitz condition, which imposes a global bound on $f'$, this condition accommodates functions with steep regions, as long as tighter matches are used locally.

To illustrate the practical advantage, consider $f(x) = x^2$ on $[0,100]$, where $|f'(x)| = 2|x|$ grows with $x$, and a global Lipschitz constant would be $L = 200$. Such a large constant makes inference difficult in finite samples. Our condition instead allows for looser matches in flatter regions and tighter matches in steeper regions—offering better practical guidance for match design.

Together, Assumption~\ref{assum:exponential-tail} (Shrinking Clusters) and Assumption~\ref{assum:derivative-control} (Derivative Control) imply:
\[
\sup_{t \in \mathcal{T}}
\left[
  \sup_{x \in \mathcal{C}_t} \bigl|f'(x)\bigr| \cdot r(\mathcal{C}_t)
\right]
\;\xrightarrow[n \to \infty]{}\; 0,
\]
under the mild requirement that $f'$ is continuous and the support of $X$ remains in a compact region. Since $f'$ is fixed (i.e., does not grow with $n$), and $r(\mathcal{C}_t) \to 0$ uniformly in $t$ by Assumption~\ref{assum:exponential-tail}, this convergence follows directly. This vanishing bound ensures that local linear approximations to $f$ within each matched cluster incur asymptotically negligible error. 

With this theoretical foundation established, we now turn to the variance estimator itself and analyze its consistency properties.

\subsection{Proposed Variance Estimator}
\label{sec:variance-estimator}

In this section, we introduce our standard error estimator for matching estimators. We begin by presenting the underlying modeling assumptions, then develop the formula for our proposed estimator, and finally establish the consistency results.

\begin{assumption}[Homoskedasticity and Regular Variance]
\label{assum:homoskedastic}
We assume that each unit has the same conditional variance under both treatment and control:
\[
\sigma_{0,i}^2 = \sigma_{1,i}^2 = \sigma_i^2.
\]
Furthermore, we assume $\sigma_i^2$ is regular in the sense of Definition~\ref{def:regular-variance}.
\end{assumption}

This homoskedasticity assumption simplifies the derivation and allows for tractable plug-in variance formulas. While the assumption may seem restrictive—since in practice the variance may differ across potential outcomes—it serves as a useful approximation, especially when matching clusters are tight—though it is important to note that even perfect matching on covariates does not imply equal error variances across potential outcomes. By assuming a single variance function $\sigma^2(x)$ governs both outcomes, we avoid needing to estimate two separate variance surfaces.

The variance $V$ consists of two components: the measurement error variance $V_E$ and the population heterogeneity variance $V_P$. We begin by developing an estimator for $V_E$, which presents more technical challenges, before extending our approach to estimate the full variance $V$. Our methodology for $V_E$ establishes key techniques that will later be applied to the full variance estimator. In both cases, our primary theoretical contribution is proving consistency of these estimators under general conditions.

\subsubsection{A Consistent Estimator for $V_E$}
To build intuition for our approach, let us first consider the special case where the variance function is constant across all covariate values, i.e., $\sigma^2(x) \equiv \sigma^2$. Under this homoskedasticity assumption, the measurement error variance simplifies to:
\begin{equation}
\begin{split}
    V_E &= \frac{1}{n_T^2}\left(\sum_{t \in \mathcal{T}} \sigma^2 + \sum_{j \in \mathcal{C}}\left(w_j\right)^2 \sigma^2\right)\\
    &= \sigma^2 \left( \frac{1}{n_T} + \frac{1}{\text{ESS}(\mathcal{C})} \right),
\end{split}
\end{equation}
where $\text{ESS}(\mathcal{C})$ is the effective sample size of the weighted control sample:
\begin{align}
    \text{ESS}(\mathcal{C}) = \frac{(\sum_{i \in \mathcal{C}} w_i)^2}{\sum_{i \in \mathcal{C}} w_i^2}. \label{eq:ESS}
\end{align}
This metric quantifies the number of independent observations that would provide equivalent precision under equal weighting \citep{potthoff_equivalent_2024}, and reflects efficiency loss from reusing controls with varying weights.

Based on this formula, our proposed plug-in estimator for $V_E$ is:
\begin{equation}
\label{eq:plug-in-estimator}
   \hat{V}_E = S^2 \left( \frac{1}{n_T} + \frac{1}{\text{ESS}(\mathcal{C})} \right),
\end{equation}
where $S^2$ is a pooled variance estimator for $\sigma^2$ defined across matched clusters where each treated unit is matched to more than one control (i.e., excluding singleton control matches). Specifically:
\begin{equation}
\label{eq:S^2}
S^2 = \frac{1}{N_C} \sum_{t \in \mathcal{T}_+} |\mathcal{C}_t| s^2_t \quad \text{with} \quad N_C = \sum_{t \in \mathcal{T}_+} |\mathcal{C}_t|,
\end{equation}
where $\mathcal{T}_+ = \{t \in \mathcal{T} : |\mathcal{C}_t| > 1\}$ excludes singleton clusters.

For each cluster, the residual variance is computed as:
\begin{equation}
\label{eq:s^2_t}
s^2_t = \frac{1}{|\mathcal{C}_t| - 1} \sum_{j \in \mathcal{C}_t} \left(Y_{j} - \bar{Y}_t\right)^2, \quad \text{where} \quad \bar{Y}_t = \frac{1}{|\mathcal{C}_t|} \sum_{j \in \mathcal{C}_t} Y_{j}.
\end{equation}

We next establish that $\hat{V}_E$ is consistent. We do this by first showing that $S^2$ consistently estimates the average treated variance, and then showing that the weights in Equation~\ref{eq:plug-in-estimator} yield an asymptotically correct representation of $V_E$, even without the homoskedasticity assumption.

\begin{lemma}[Consistency of the Pooled Variance Estimator]
\label{lemma:consistency-homogeneous}
Let $\{\mathcal{C}_t, t\in \mathcal{T}\}$ be a collection of matched control sets. Assume Assumptions~\ref{assum:exponential-tail} (Shrinking Clusters), \ref{assum:regular-variance} (Regular Variance), and \ref{assum:derivative-control} (Derivative Control). Then, as $n_T \rightarrow \infty$:
\begin{equation}
\label{eq:consistency-homogeneity}
    \left|S^2 - \frac{1}{n_T} \sum_{t=1}^{n_T} \sigma_t^2\right| \xrightarrow{a.s.} 0.
\end{equation}
\end{lemma}

\textit{Proof:} See Appendix~\ref{sec:proof-lemma:consistency-homogeneous}.

This lemma establishes two key properties of our pooled variance estimator. First, although it was initially motivated under the homoskedasticity assumption, the estimator $S^2$ remains consistent even when error variances $\sigma_t^2$ are heterogeneous across units. This robustness arises because local variance estimates from matched clusters are close to the true $\sigma^2(X_t)$ due to the shrinking clusters property and regularity of the variance function. We also note that the proportion of treated units matched to only a single control becomes negligible as $n_T \to \infty$, which justifies focusing on the subset $\mathcal{T}_+$ when computing $S^2$.

Second, while individual cluster variance estimates $s_t^2$ may be noisy or unreliable when computed from small matched sets, the pooled estimator $S^2$ provides a good estimate by aggregating information across many clusters. This aggregation smooths out local estimation errors, mirroring the robustness principle underlying White's heteroskedasticity-consistent variance estimator in regression contexts.

To show that $\hat{V}_E$ remains valid even without Assumption~\ref{assum:homoskedastic}, we next demonstrate that the true variance $V_E$ asymptotically depends only on the treated-unit variances. This reduces the reliance on the homoskedasticity assumption and sets up the justification for using $S^2$ as a proxy.

\begin{lemma}[Asymptotic Equivalence to Error Variance]
\label{lemma:variance-equivalence}
Under the same assumptions as Lemma~\ref{lemma:consistency-homogeneous}, define:
\[
\hat{V}_{E,\text{lim}} := \left( \frac{1}{n_T} \sum_{t=1}^{n_T} \sigma_t^2 \right) \left( \frac{1}{n_T} + \frac{1}{\text{ESS}(\mathcal{C})} \right).
\]
Then:
\[
\left|\hat{V}_{E,\text{lim}} - V_E \right| \xrightarrow{p} 0 \quad \text{as} \; n_T \rightarrow \infty.
\]
\end{lemma}

\textit{Proof:} See Appendix~\ref{sec:proof-lemma:variance-equivalence}.

This lemma shows that our variance formula (Equation~\ref{eq:plug-in-estimator}) resembles the pooled variance structure in a two-sample t-test assuming equal variances. However, the technical details are non-trivial. Note that $\hat{V}_{E,\text{lim}}$ depends on only the treated unit variances, while the true variance $V_E$ includes the weighted average of control unit variances through the term $\frac{1}{n_T^2} \sum_{j \in \mathcal{C}} w_j^2 \sigma_j^2$.

The convergence $V_{E,\text{lim}} - V_E \xrightarrow{P} 0$ demonstrates that two competing forces reach perfect balance: the \textbf{concentration effect} of weights ($\sum_{j \in \mathcal{C}} w_j^2$) and the \textbf{averaging effect} through effective sample size ($\frac{1}{\text{ESS}(\mathcal{C})}$). When these forces balance, the control variance contribution $\frac{1}{n_T^2} \sum_{j \in \mathcal{C}} w_j^2 \sigma_j^2$ in $V_E$ becomes asymptotically equivalent to the reweighted average of the treated-unit variances $\frac{1}{n_T} \frac{1}{\text{ESS}(\mathcal{C})} \sum_{t=1}^{n_T} \sigma_t^2$ in $\hat{V}_{E,\text{lim}}$.

This variance-weighting equilibrium connects to the broader balancing literature by showing that effective matching procedures achieve \textbf{asymptotic variance balance}: the ESS naturally emerges as both a concentration penalty and a balancing diagnostic that quantifies how evenly the weighting distributes across control units. The equilibrium condition provides guidance for optimal weight selection and suggests that good matching naturally balances variance contributions across treatment and control groups.

A second key observation is that the control unit weights $w_j$ appear in our variance estimator only through the $\text{ESS}(\mathcal{C})$ term. This provides clean insight into how weighting schemes contribute to variance in matching and weighting estimators more generally, as we explore further in Section~\ref{sec:application-to-others}.

Combining these two lemmas, we now show that the plug-in estimator $\hat{V}_E$ is consistent for the true variance $V_E$.

\begin{theorem}[Consistency of the Variance Estimator]
\label{thm:consistency-variance}
Under Assumptions~\ref{assum:exponential-tail}, \ref{assum:regular-variance}, and \ref{assum:derivative-control}, the proposed estimator $\hat{V}_E$ is consistent:
\[
\left|\hat{V}_E - V_E\right| \xrightarrow{p} 0 \quad \text{as} \; n_T \rightarrow \infty.
\]
\end{theorem}

\begin{proof}
From Lemma \ref{lemma:consistency-homogeneous}, we have $\left|S^2 - \frac{1}{n_T} \sum_{t=1}^{n_T} \sigma_t^2\right| \xrightarrow{a.s.} 0$. Substituting into our estimator formula:
\begin{align*}
\hat{V}_E &= S^2 \left( \frac{1}{n_T} + \frac{1}{\text{ESS}(\mathcal{C})} \right) \\
&= \left(\frac{1}{n_T} \sum_{t=1}^{n_T} \sigma_t^2 + o_p(1)\right) \left( \frac{1}{n_T} + \frac{1}{\text{ESS}(\mathcal{C})} \right) \\
&= \hat{V}_{E, \text{lim}} + o_p(1)
\end{align*}

By Lemma \ref{lemma:variance-equivalence}, we have $\left|\hat{V}_{E, \text{lim}}-V_E \right| \, \xrightarrow{p} 0$. Therefore:
\begin{align*}
\left|\hat{V}_E - V_E\right| &= \left|\hat{V}_{E, \text{lim}} + o_p(1) - V_E\right| \\
&\leq \left|\hat{V}_{E, \text{lim}} - V_E\right| + \left|o_p(1)\right| \\
&\xrightarrow{p} 0
\end{align*}
\end{proof}

Theorem~\ref{thm:consistency-variance} provides the key theoretical guarantee of our method: a plug-in variance estimator motivated by homoskedasticity remains consistent even under general heteroskedastic error structures, as long as the regularity conditions are met. Our estimator offers practical advantages in high-overlap settings, where reuse of control units inflates variance—an effect directly captured by the ESS term.

Our non-parametric approach differs from Theorem 1 of \cite{white1980heteroskedasticity}, which uses a regression-based (semi-parametric) method. While our matching procedure is governed by hyperparameters such as the number of neighbors or the maximum allowed radius, these parameters are not estimated from the data. Consequently, we require Assumption~\ref{assum:regular-variance} (Regular Variance), especially the continuity condition in Definition~\ref{def:regular-variance}, whereas \cite{white1980heteroskedasticity} does not need such an assumption. While both proofs share the same overall strategy, the specific technical details differ: \cite{white1980heteroskedasticity}'s argument relies on compactness of the parameter space to bound the difference between the estimator and the truth, whereas we rely on Assumptions~\ref{assum:exponential-tail} (Shrinking Clusters) and~\ref{assum:derivative-control} (Derivative Control). Further details on this comparison can be found in Appendix~\ref{sec:comparison-to-white-HC}.

\subsubsection{A Consistent Estimator for $V$}
Building on our analysis of the measurement error variance component $V_E$, we now develop a consistent estimator for the total variance $V$. While $V_E$ captures the variance due to residual outcome noise, the complete variance $V$ must also account for treatment effect heterogeneity among the treated units.

We start by exploring the relationship between the squared deviations of individual treatment effects and the components of the total variance:

$$
\begin{aligned}
& E\left[\left(Y_t(1)-\hat{Y}_t(0)-\tau\right)^2\right] \\
\approx & E\left[(\tau(x)-\tau)^2\right]+E\left[\varepsilon_t^2+\sum_{j \in C_t} w_{j t}^2 \varepsilon_j^2\right] \\
\approx & \frac{1}{n_T} V_P + \frac{1}{n_T}\left[\sum_{t \in \mathcal{T}} \sigma_t^2+\sum_{j \in \mathcal{C}}\left(\sum_{t' \in \mathcal{T}} w_{j t'}^2\right) \sigma_j^2\right]
\end{aligned}
$$

This expectation can be approximated empirically as:

$$
\begin{aligned}
E\left[\left(Y_t(1)-\hat{Y}_t(0)-\tau\right)^2\right] \approx \frac{1}{n_T} \sum_{t \in \mathcal{T}}\left(Y_t-\hat{Y}_t(0)-\hat{\tau}\right)^2
\end{aligned}
$$

By equating these expressions and rearranging terms, we can derive an estimator for $V_P$:

$$
\begin{aligned}
\hat{V}_P \approx & \frac{1}{n_T} \sum_{t \in \mathcal{T}}\left(Y_t-\hat{Y}_t(0)-\hat{\tau}\right)^2 \\
& -\frac{1}{n_T}\left[\sum_{t \in \mathcal{T}} \hat{\sigma}_t^2+\sum_{j \in \mathcal{C}}\left(\sum_{t' \in \mathcal{T}} w_{j t'}^2\right) \hat{\sigma}_j^2\right]
\end{aligned}
$$

We now heuristically combine the two components. While this expression includes variance estimates $\hat{\sigma}_t^2$ and $\hat{\sigma}_j^2$, the former terms will cancel out and the latter will ultimately be approximated by $S^2$ as in our estimator for $V_E$. One does not need to worry about the precise form of these terms at this stage—they serve to motivate the algebraic derivation below.

Combining this with our estimator for $V_E$, we obtain:

$$
\begin{aligned}
\hat{V}= & n_T \cdot (\hat{V}_E + \hat{V}_P) \\
= & \frac{1}{n_T}\left[\sum_{t \in \mathcal{T}} \hat{\sigma}_t^2+\sum_{j \in \mathcal{C}}\left(\sum_{t' \in \mathcal{T}} w_{j t'}\right)^2 \hat{\sigma}_j^2\right] \\
& +\frac{1}{n_T} \sum_{t \in \mathcal{T}}\left(Y_t-\hat{Y}_t(0)-\hat{\tau}\right)^2 \\
& -\frac{1}{n_T}\left[\sum_{t \in \mathcal{T}} \hat{\sigma}_t^2+\sum_{j \in \mathcal{C}}\left(\sum_{t' \in \mathcal{T}} w_{j t'}^2\right) \hat{\sigma}_j^2\right]
\end{aligned}
$$

Through algebraic simplification, this expression reduces to:

$$
\begin{aligned}
\hat{V} = & \frac{1}{n_T} \sum_{t \in \mathcal{T}}\left(Y_t-\hat{Y}_t(0)-\hat{\tau}\right)^2 \\
& +S^2 \frac{1}{n_T}\left[\sum_{j \in \mathcal{C}}\left[\left(\sum_{t' \in \mathcal{T}} w_{j t'}\right)^2-\left(\sum_{t' \in \mathcal{T}} w_{j t'}^2\right)\right]\right]
\end{aligned}
$$

where $S^2$ is the pooled variance defined in Equation~\ref{eq:S^2}. This estimator effectively combines the empirical squared deviations with a correction term that accounts for the matching structure.

\begin{theorem}[Consistency of the Total Variance Estimator]
\label{thm:consistency-total-variance}
Under Assumptions~\ref{assum:exponential-tail}, \ref{assum:regular-variance}, and \ref{assum:derivative-control}, the proposed estimator $\hat{V}$ is consistent:
\[
\left|\hat{V} - V\right| \xrightarrow{p} 0 \quad \text{as} \; n_T \rightarrow \infty.
\]
\end{theorem}

The proof follows similar steps to those used in establishing the consistency of $\hat{V}_E$ in Theorem~\ref{thm:consistency-variance}. The key insight is that both the empirical squared deviations and the correction term converge to their respective population counterparts in probability, leveraging the properties of shrinking clusters, regular error variance functions, and our derivative control condition.

This consistency result ensures that confidence intervals constructed using $\hat{V}$ will have asymptotically correct coverage, providing practitioners with reliable inference tools for matching estimators across a wide range of applications.

\subsection{Comparison with \cite{abadie2006large} Estimator}
\label{sec:comparison}

To position our work within the existing literature and highlight its advantages, we now compare our variance estimator with that proposed by \cite{abadie2006large}. This comparison is particularly relevant as their work established the foundational theory for matching estimators, and our analysis builds upon and extends their approach for practical applications in modern causal inference settings.

Adapting their estimator to our notation:
\begin{align}
\label{eq:AI06-plug-in-variance-estimator}
   \widehat{V}_{AI06} &=
        \frac{1}{n_T^2} \sum_{t \in \mathcal{T}} \hat \sigma_t^2 +
        \frac{1}{n_T^2} \sum_{j \in \mathcal{C}} \left(\sum_{t \in \mathcal{T}} w_{jt}\right)^2 \hat \sigma_j^2,
\end{align}
where $w_{jt} = 1/M$ if unit $j$ is among the $M$ closest controls to unit $t$, and $w_{jt} = 0$ otherwise, and $\hat \sigma_i^2$ is an estimate of the conditional outcome variance for unit $i$, defined as:
\[
\hat\sigma_i^2 = \frac{M}{M+1} \left(Y_i - \frac{1}{M} \sum_{m=1}^{M} Y_{m(i)}\right)^2.
\]

Here, $Y_{m(i)}$ denotes the outcome of the $m$-th closest unit to unit $i$ among units with the same treatment status, and $M$ is a fixed small number (typically set to match the number of matches used in the estimator).

The fundamental methodological difference lies in variance estimation approaches. \cite{abadie2006large} estimates variance by comparing each unit to its nearest same-treatment neighbors individually: $\hat\sigma_i^2 = \frac{M}{M+1} \left(Y_i - \frac{1}{M} \sum_{m=1}^{M} Y_{m(i)}\right)^2$. In contrast, our estimator calculates variance within matched clusters: $s^2_t = \frac{1}{|\mathcal{C}_t| - 1} \sum_{j \in \mathcal{C}_t} \left(Y_{j} - \bar{Y}_t\right)^2$, pooling information across all controls matched to each treated unit.

This difference in approach leads to several important practical advantages and trade-offs. First, our estimator requires only matching controls to treated units, whereas \cite{abadie2006large} requires matching for both treatment and control groups—significantly reducing computational burden when the control group is large. However, this computational advantage comes at the cost of requiring our homoskedasticity assumption ($\sigma_t^2 = \sigma_c^2$ for matched pairs), while \cite{abadie2006large} can accommodate arbitrary heteroskedasticity across units. 

Second, \cite{abadie2006large}'s approach necessitates matching treated units with other treated units to estimate $\hat\sigma_t^2$. This becomes problematic when the treated group is small or highly heterogeneous in covariates, as finding good same-treatment matches becomes difficult or impossible. Our approach avoids this issue entirely by focusing on control-to-treated matching, making it particularly suitable for ATT estimation where treated samples are typically small.

Third, our framework naturally accommodates flexible weighting schemes, including kernel weights, caliper matching weights, and optimal transportation weights, whereas \cite{abadie2006large}'s approach was primarily designed for fixed-number nearest neighbor matching with equal weights.

The main limitation of our approach is that we do not utilize within-treated-group variation for variance estimation—we do not use the observed outcomes $Y_t$ of treated units when estimating $\sigma^2$, potentially discarding valuable information. This efficiency loss is the price of our computational simplicity and homoskedasticity assumption. However, this limitation is typically minor in ATT applications where the treated group is small relative to the control group, and within-treated-group variation becomes unreliable when the number of treated units is small. Many influential ATT applications feature relatively small treated samples, including job training program evaluations \citep{lalonde1986evaluating}, educational interventions \citep{abadie2002instrumental}, and health policy assessments \citep{keele_hospital_2023}, where \cite{imbens2004nonparametric} notes that ATT estimation is often preferred precisely because treatment is relatively rare or targeted.

\section{Simulation}
In this section, we conduct simulation studies to validate the two main theoretical results established in earlier sections: Theorem \ref{thm:clt} (Central Limit Theorem) and the consistency of our variance estimator. The primary focus is threefold: first, to verify the asymptotic normality of our estimator, second, to assess whether confidence intervals constructed using our variance estimator achieve near-nominal coverage, and third, to compare the performance of our variance estimator to that of existing methods, demonstrating how our approach substantially outperforms the state-of-the-art bootstrap variance estimator proposed by \cite{otsu2017bootstrap}. These simulations provide empirical insights into the reliability and robustness of our methods under different data-generating scenarios and matching conditions.

\subsection{Otsu-Rai DGP: Challenging Nonlinear Setting}

We first evaluate our method using a challenging simulation design from \cite{otsu2017bootstrap} featuring a complex nonlinear outcome function in two dimensions. The data generating process is defined as:
\begin{align*} 
&\left\{Y_i, Z_i, \mathbf{X}_i\right\}_{i=1}^n, \\ 
& Y_i(1)=\tau+m\left(\left\|\mathbf{X}_i\right\|\right)+\epsilon_i, \quad Y_i(0)=m\left(\left\|\mathbf{X}_i\right\|\right)+\epsilon_i, \\ 
& Z_i=\mathbb{I}\left\{P\left(\mathbf{X}_i\right) \geq v_i\right\}, \quad v_i \sim U[0,1], \\ 
& P\left(\mathbf{X}_i\right)=\gamma_1+\gamma_2\left\|\mathbf{X}_i\right\|, \quad \mathbf{X}_i=\left(X_{1 i}, X_{2 i}\right)^{\prime}, \\ 
& X_{j i}=\xi_i\left|\zeta_{j i}\right| /\left\|\boldsymbol{\zeta}_i\right\| \quad \text { for } j=1, 2, \\ 
& \xi_i \sim U[0,1], \quad \boldsymbol{\zeta}_i \sim N\left(\mathbf{0}, I_2\right), \quad \epsilon_i \sim N\left(0,0.2^2\right),
\end{align*}
where $\gamma_1 = 0.15$, $\gamma_2 = 0.7$, $\tau = 0$, and the nonlinear outcome function is $m(z) = 0.4 + 0.25\sin(8z-5) + 0.4\exp(-16(4z-2.5)^2)$ (curve 6 from \cite{otsu2017bootstrap}). We implement 8-nearest neighbor matching with uniform weighting ($w_{jt} = 1/8$) across 100 replications with approximately $n_T = n_C = 50$.

To quantify the dependency structure created by matching, we measure \textit{control unit reuse}—the phenomenon where the same control units are matched to multiple treated units. Specifically, we calculate for each treated unit how many of its matched controls are shared with any other treated unit, which directly affects the reliability of that unit's counterfactual estimate and the validity of bootstrap resampling procedures.

Table~\ref{tab:otsu-comparison} presents our main empirical findings. The results reveal a dramatic performance gap between methods: our pooled variance estimator achieves 97\% coverage, very close to the nominal 95\% rate, while the wild bootstrap method proposed by \cite{otsu2017bootstrap} achieves only 61\% coverage—a severe undercoverage that would lead to highly misleading inference in practice.

\begin{table}[ht]
\centering
\caption{Otsu-Rai DGP: Performance Comparison}
\label{tab:otsu-comparison}
\begin{tabular}{lccc}
\toprule
\textbf{Method} & \textbf{Coverage (\%)} & \textbf{Avg CI Length} & \textbf{Control Reuse Pattern} \\
\midrule
Wild Bootstrap & 61.0 & 0.156 & 23 shared controls per treated unit \\
Our Method & 97.0 & 0.313 & 23 shared controls per treated unit \\
\midrule
\textit{Target} & \textit{95.0} & \textit{0.238} & \textit{—} \\
\bottomrule
\end{tabular}
\vspace{0.1in}
\end{table}

The poor performance of the wild bootstrap can be attributed to the substantial control unit reuse in the matching structure. Each treated unit shares an average of 23 controls with other treated units, creating complex dependencies that bootstrap resampling fails to capture properly. In contrast, our method explicitly accounts for this dependency structure through the effective sample size calculation and pooled variance estimation, maintaining proper coverage even under extensive control reuse. The slightly conservative coverage (97\% vs 95\%) reflects the challenging nature of this DGP, where the nonlinear outcome function creates additional complexity that our method handles robustly.

To confirm that control unit reuse drives the performance difference, we conducted an additional experiment reducing reuse by setting $n_T = 25$ and $n_C = 1000$. With minimal reuse (average shared controls dropping to 0.14 per treated unit), the wild bootstrap recovered nominal coverage (95\%) while our method maintained robust performance (94\% coverage), supporting our hypothesis that bootstrap failure occurs specifically under substantial control reuse conditions.

\subsection{Che et al. DGP: Multi-Dimensional Validation}

To provide comprehensive validation of our theoretical framework, we conduct additional simulations following the design from \cite{che2024caliper}. This four-dimensional setting with varying degrees of population overlap provides secondary evidence of our method's robustness across different scenarios.

We simulate one hundred treated units and five hundred control units, with covariates drawn from a 4-dimensional multivariate normal distribution $N((0.5, 0.5, 0.5, 0.5)^T, \Sigma)$, where the covariance matrix $\Sigma$ has diagonal elements $\Sigma_{ii} = 1$ for all $i$ and off-diagonal elements $\Sigma_{ij} = 0.8$ for $i \neq j$. For each unit with covariates $(x_1, x_2, x_3, x_4)$, we generate outcomes via $Y = f_0(x_1, x_2, x_3, x_4) + Z \cdot \tau(x_1, x_2, x_3, x_4) + \epsilon$, where $f_0$ is the density function for the same multivariate normal distribution, $\tau(x_1, x_2, x_3, x_4) = 3 \sum_{i=1}^4 x_i$ is the heterogeneous treatment effect function, and $\epsilon \sim N(0, 0.5^2)$ represents homogeneous measurement error. We vary the degree of overlap by adjusting the distribution parameters and use synthetic-control-like optimization within local neighborhoods with adaptive caliper sizes.

We first verify the asymptotic normality predicted by our Central Limit Theorem (Theorem~\ref{thm:clt}) by plotting the distribution of $\frac{\sqrt{n_T} \,\bigl(\hat{\tau} - B_n - \tau\bigr)}{\hat V^{-1/2}}$ against a standard normal distribution in Figure~\ref{fig:clt-verification}. The close alignment between these distributions provides intuitive evidence that our estimator's distribution converges to the theoretical limit, confirming the accuracy of our CLT result.

\begin{figure}[ht]
    \centering
    \includegraphics[width=0.8\linewidth]{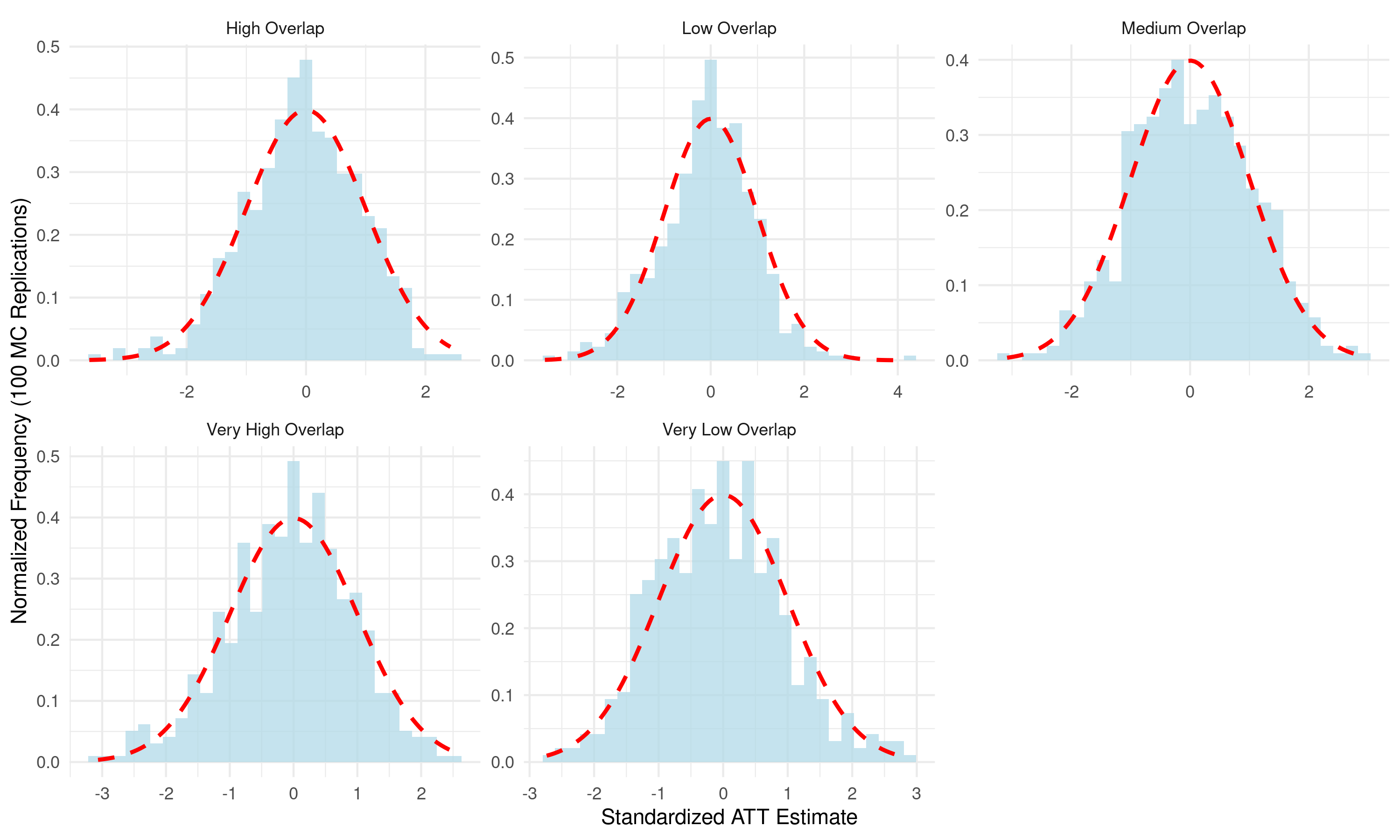}
    \caption{Empirical distribution of 
    $\frac{\sqrt{n_T} \,\bigl(\hat{\tau} - B_n - \tau\bigr)}{\hat V^{1/2}}$ 
    versus the standard normal distribution. The x-axis shows the standardized values and the y-axis shows the frequency density from 100 Monte Carlo replications.}
    \label{fig:clt-verification}
\end{figure}

Table~\ref{tab:che-validation} presents coverage performance across different overlap scenarios. Our method consistently outperforms bootstrap across all settings, with coverage rates near the nominal 95\% level while bootstrap shows persistent undercoverage.

\begin{table}[ht]
\centering
\caption{Che et al. DGP: Coverage Performance Across Overlap Scenarios}
\label{tab:che-validation}
\begin{tabular}{lccccc}
\toprule
& \multicolumn{2}{c}{\textbf{Bootstrap}} & \multicolumn{2}{c}{\textbf{Our Method}} & \\
\cmidrule(lr){2-3} \cmidrule(lr){4-5}
\textbf{Overlap Level} & \textbf{Coverage} & \textbf{Avg $\widehat{SE}$} & \textbf{Coverage} & \textbf{Avg $\widehat{SE}$} & \textbf{Control Reuse} \\
\midrule
Very Low & 90.8\% & 0.308 & 95.0\% & 0.358 & 11.9 units \\
Low & 92.2\% & 0.310 & 94.6\% & 0.348 & 10.8 units \\
Medium & 92.4\% & 0.310 & 94.0\% & 0.341 & 9.7 units \\
High & 93.4\% & 0.309 & 94.4\% & 0.334 & 8.2 units \\
Very High & 92.6\% & 0.309 & 94.4\% & 0.331 & 6.2 units \\
\bottomrule
\end{tabular}
\\[0.1cm]
\end{table}

Our pooled variance estimator demonstrates superior performance across all overlap scenarios, with coverage rates consistently near the nominal 95\% level (ranging from 94.0\% to 95.0\%). The bootstrap method shows systematically lower coverage rates, ranging from 90.8\% to 93.4\%, indicating persistent undercoverage. The moderate degree of control reuse in this setting (6-12 units) compared to the extreme reuse in the Otsu-Rai setting (23 units) explains why the performance gap is smaller but still consistent.

As expected, both methods perform better as overlap increases due to improved matching quality and reduced control reuse. However, even in the most favorable scenario (very high overlap), the bootstrap method fails to achieve nominal coverage, while our method consistently maintains proper inference properties.

\subsection{Summary of Simulation Evidence}

The simulation evidence provides strong support for our theoretical framework and demonstrates the practical importance of our contributions. In challenging settings with substantial control unit reuse—a common occurrence in real-world matching applications—our method maintains proper coverage while the current state-of-the-art bootstrap approach can fail dramatically. The robustness of our approach across different data generating processes, dimensions, and overlap patterns makes it a reliable tool for practitioners seeking valid population inference in matching-based causal studies.

Additional results examining the performance of individual variance components and bias correction effects are presented in Appendix~\ref{sec:additional-simulation-results}.

\section{Extending the Variance Estimator to Other Estimators}
\label{sec:application-to-others}

While our variance estimator was developed in the context of matching methods, its utility extends to other classes of estimators that share similar properties. In this section, we demonstrate how our approach can be applied to weighting estimators, specifically the stable balancing weights method proposed by \cite{zubizarreta2015stable}.

The stable balancing weights approach of \cite{zubizarreta2015stable} finds weights that minimize the variance of the weighted estimator while satisfying covariate balance constraints. Using our notation, the stable balancing weights estimator for the ATT can be expressed as:

\begin{equation}
\hat{\tau}_{SBW} = \frac{1}{n_T} \sum_{t \in \mathcal{T}} Y_t - \frac{1}{n_T} \sum_{j \in \mathcal{C}} w_j Y_j
\end{equation}

where $w_j$ are weights assigned to control units that minimize $\sum_{j \in \mathcal{C}} w_j^2$ subject to balance constraints of the form $\left| \sum_{t \in \mathcal{T}} \frac{1}{n_T} X_t - \sum_{j \in \mathcal{C}} w_j X_j \right| \leq \delta$ for some small tolerance $\delta$.

To make a variance estimator for the stable balancing weights, we can directly extend our framework. We can apply our estimator from Equation~\ref{eq:plug-in-estimator} with only one modification: while the weights $w_j$ are determined through quadratic optimization rather than matching, the construction of the heteroskedastic variance component $S^2$ still requires forming local neighborhoods through matching. This hybrid approach leverages the computational advantages of both techniques—optimal weights from the balancing procedure and accurate variance estimation from local matching—resulting in valid inference for the weighting estimator.

To evaluate the performance of our variance estimator when applied to the stable balancing weights estimator, we conducted a simulation study using the data generating process (DGP) proposed by \cite{kang2007demystifying}. This DGP is widely used in the causal inference literature as a challenging benchmark due to its non-linear relationships between covariates, treatment, and outcomes.

The Kang and Schafer DGP generates four standard normal covariates $(X_1, X_2, X_3, X_4)$ and then creates non-linear transformations to produce observed covariates. The treatment assignment is a function of these covariates, and the outcome model includes interactions between treatment and covariates, creating a complex setting where many estimators struggle to achieve proper coverage. Full mathematical details of this DGP are provided in Appendix~\ref{subsec:ks-dgp}.

We applied the stable balancing weights estimator with our variance estimation approach to the Kang and Schafer DGP over 100 independent replications. The 95\% confidence intervals constructed using our variance estimator achieved a coverage rate of 98\%. This slightly conservative coverage indicates that our variance estimator remains valid even when applied to weighting estimators in challenging settings.

The over-coverage can be attributed to the steep gradient in the Kang and Schafer outcome model, which likely violates our derivative control condition (Assumption~\ref{assum:derivative-control}). In such cases, the product of the local derivative magnitude and cluster radius may exceed the required bound, leading to larger effective bias terms that are not fully captured by our first-order approximation. This suggests directions for future research, particularly in developing refined variance estimators that can better handle violations of the derivative control condition. Potential approaches could include higher-order bias corrections or adaptive methods that estimate the local curvature of the outcome model to adjust the variance estimation accordingly.

Nevertheless, the strong performance of our variance estimator when applied to the stable balancing weights approach demonstrates its flexibility and robustness. This extension opens possibilities for creating a unified framework for inference across various classes of weighting and matching estimators in causal inference.

\section{Conclusion}
This paper develops new methods for statistical inference in matching estimators, addressing key challenges in both theoretical foundations and practical implementation. We propose a computationally efficient variance estimator that remains valid under extensive control unit reuse—a critical advantage over existing approaches that either require computationally prohibitive matching for both treatment groups or fail under realistic conditions with control unit dependencies. Our theoretical framework introduces novel regularity conditions that significantly expand the applicability of matching methods: the derivative control condition allows outcome functions to have varying rates of change across the covariate space as long as this variation is balanced by appropriate cluster sizes, while our shrinking clusters assumption provides flexibility by requiring only eventual convergence without specifying particular rates.

Through simulation studies, we demonstrate that our variance estimator achieves substantially more reliable inference than existing methods, particularly in settings with extensive control unit reuse. Most notably, our method maintains 97\% coverage while the state-of-the-art wild bootstrap approach drops to 61\% coverage in challenging but realistic scenarios—a performance gap that highlights critical limitations of current methods when the same control units are matched to multiple treated units. This robustness stems from our explicit modeling of dependency structures through effective sample size calculations and pooled variance estimation, allowing our method to maintain accurate coverage even in high-overlap settings where bootstrap methods systematically underestimate uncertainty.

Our approach provides a practical and theoretically grounded solution for researchers conducting inference in matched designs, with immediate applications in economics, epidemiology, and policy evaluation where valid population inference is essential for policy decisions. The method's computational efficiency, theoretical rigor, and empirical robustness make it particularly valuable for average treatment effect on the treated estimation with small treated samples. Future work will explore extensions to high-dimensional covariate spaces, adaptive bias correction techniques that adjust to local outcome function properties, and applications to longitudinal matching designs with time-varying treatment effects. We also plan to investigate unified variance estimation frameworks that work across both matching and weighting approaches, building on our successful extension to stable balancing weights.

\bibliographystyle{plainnat}
\bibliography{refs}

\begin{thebibliography}{27}
\providecommand{\natexlab}[1]{#1}
\providecommand{\url}[1]{\texttt{#1}}
\expandafter\ifx\csname urlstyle\endcsname\relax
  \providecommand{\doi}[1]{doi: #1}\else
  \providecommand{\doi}{doi: \begingroup \urlstyle{rm}\Url}\fi

\bibitem[Abadie and Imbens(2006)]{abadie2006large}
Alberto Abadie and Guido~W Imbens.
\newblock Large sample properties of matching estimators for average treatment effects.
\newblock \emph{Econometrica}, 74\penalty0 (1):\penalty0 235--267, 2006.

\bibitem[Abadie and Imbens(2008)]{abadie2008failure}
Alberto Abadie and Guido~W Imbens.
\newblock On the failure of the bootstrap for matching estimators.
\newblock \emph{Econometrica}, 76\penalty0 (6):\penalty0 1537--1557, 2008.

\bibitem[Abadie and Imbens(2011)]{abadie2011bias}
Alberto Abadie and Guido~W Imbens.
\newblock Bias-corrected matching estimators for average treatment effects.
\newblock \emph{Journal of Business \& Economic Statistics}, 29\penalty0 (1):\penalty0 1--11, 2011.

\bibitem[Abadie and Imbens(2012)]{abadie2012martingale}
Alberto Abadie and Guido~W Imbens.
\newblock A martingale representation for matching estimators.
\newblock \emph{Journal of the American Statistical Association}, 107\penalty0 (498):\penalty0 833--843, 2012.

\bibitem[Abadie et~al.(2002)Abadie, Angrist, and Imbens]{abadie2002instrumental}
Alberto Abadie, Joshua Angrist, and Guido Imbens.
\newblock Instrumental variables estimates of the effect of subsidized training on the quantiles of trainee earnings.
\newblock \emph{Econometrica}, 70\penalty0 (1):\penalty0 91--117, 2002.

\bibitem[Che et~al.(2024)Che, Meng, and Miratrix]{che2024caliper}
Jonathan Che, Xiang Meng, and Luke Miratrix.
\newblock Caliper synthetic matching: Generalized radius matching with local synthetic controls.
\newblock \emph{arXiv preprint arXiv:2411.05246}, 2024.

\bibitem[Dehejia and Wahba(1999)]{dehejia1999causal}
Rajeev~H Dehejia and Sadek Wahba.
\newblock Causal effects in nonexperimental studies: Reevaluating the evaluation of training programs.
\newblock \emph{Journal of the American Statistical Association}, 94\penalty0 (448):\penalty0 1053--1062, 1999.

\bibitem[Hahn et~al.(2020)Hahn, Murray, and Carvalho]{hahn2020bayesian}
P~Richard Hahn, Jared~S Murray, and Carlos~M Carvalho.
\newblock Bayesian regression tree models for causal inference: Regularization, confounding, and heterogeneous effects (with discussion).
\newblock \emph{Bayesian Analysis}, 15\penalty0 (3):\penalty0 965--1056, 2020.

\bibitem[Heckman et~al.(1997)Heckman, Ichimura, and Todd]{heckman1997matching}
James~J Heckman, Hidehiko Ichimura, and Petra~E Todd.
\newblock Matching as an econometric evaluation estimator: Evidence from evaluating a job training programme.
\newblock \emph{The Review of Economic Studies}, 64\penalty0 (4):\penalty0 605--654, 1997.

\bibitem[Heckman et~al.(1998)Heckman, Ichimura, and Todd]{heckman1998matching}
James~J Heckman, Hidehiko Ichimura, and Petra Todd.
\newblock Matching as an econometric evaluation estimator.
\newblock \emph{The Review of Economic Studies}, 65\penalty0 (2):\penalty0 261--294, 1998.

\bibitem[Hill(2011)]{hill2011bayesian}
Jennifer~L Hill.
\newblock Bayesian nonparametric modeling for causal inference.
\newblock \emph{Journal of Computational and Graphical Statistics}, 20\penalty0 (1):\penalty0 217--240, 2011.

\bibitem[Hirano et~al.(2003)Hirano, Imbens, and Ridder]{hirano2003efficient}
Keisuke Hirano, Guido~W Imbens, and Geert Ridder.
\newblock Efficient estimation of average treatment effects using the estimated propensity score.
\newblock \emph{Econometrica}, 71\penalty0 (4):\penalty0 1161--1189, 2003.

\bibitem[Imbens(2004)]{imbens2004nonparametric}
Guido~W Imbens.
\newblock Nonparametric estimation of average treatment effects under exogeneity: A review.
\newblock \emph{Review of Economics and Statistics}, 86\penalty0 (1):\penalty0 4--29, 2004.

\bibitem[Kallus(2020)]{kallus2020generalized}
Nathan Kallus.
\newblock Generalized optimal matching methods for causal inference.
\newblock \emph{J. Mach. Learn. Res.}, 21:\penalty0 62--1, 2020.

\bibitem[Kang and Schafer(2007)]{kang2007demystifying}
Joseph~DY Kang and Joseph~L Schafer.
\newblock Demystifying double robustness: A comparison of alternative strategies for estimating a population mean from incomplete data.
\newblock 2007.

\bibitem[Keele et~al.(2023)Keele, Ben-Michael, Feller, Kelz, and Miratrix]{keele_hospital_2023}
Luke~J. Keele, Eli Ben-Michael, Avi Feller, Rachel Kelz, and Luke Miratrix.
\newblock Hospital quality risk standardization via approximate balancing weights.
\newblock \emph{The Annals of Applied Statistics}, 17\penalty0 (2), June 2023.
\newblock ISSN 1932-6157.
\newblock \doi{10.1214/22-AOAS1629}.
\newblock URL \url{https://projecteuclid.org/journals/annals-of-applied-statistics/volume-17/issue-2/Hospital-quality-risk-standardization-via-approximate-balancing-weights/10.1214/22-AOAS1629.full}.

\bibitem[LaLonde(1986)]{lalonde1986evaluating}
Robert~J LaLonde.
\newblock Evaluating the econometric evaluations of training programs with experimental data.
\newblock \emph{The American economic review}, pages 604--620, 1986.

\bibitem[Otsu and Rai(2017)]{otsu2017bootstrap}
Taisuke Otsu and Yoshiyasu Rai.
\newblock Bootstrap inference of matching estimators for average treatment effects.
\newblock \emph{Journal of the American Statistical Association}, 112\penalty0 (520):\penalty0 1720--1732, 2017.

\bibitem[Potthoff et~al.(2024)Potthoff, Woodbury, and Manton]{potthoff_equivalent_2024}
Richard~F Potthoff, Max~A Woodbury, and Kenneth~G Manton.
\newblock "{Equivalent} {Sample} {Size}" and "{Equivalent} {Degrees} of {Freedom}" {Refinements} for {Inference} {Using} {Survey} {Weights} {Under} {Superpopulation} {Models}.
\newblock 2024.

\bibitem[Rosenbaum and Rubin(1983)]{rosenbaum1983central}
Paul~R Rosenbaum and Donald~B Rubin.
\newblock The central role of the propensity score in observational studies for causal effects.
\newblock \emph{Biometrika}, 70\penalty0 (1):\penalty0 41--55, 1983.

\bibitem[Rubin(1973)]{rubin1973matching}
Donald~B Rubin.
\newblock Matching to remove bias in observational studies.
\newblock \emph{Biometrics}, pages 159--183, 1973.

\bibitem[Rubin(1974)]{rubin1974estimating}
Donald~B Rubin.
\newblock Estimating causal effects of treatments in randomized and nonrandomized studies.
\newblock \emph{Journal of educational Psychology}, 66\penalty0 (5):\penalty0 688, 1974.

\bibitem[Smith and Todd(2005)]{smith2005does}
Jeffrey~A Smith and Petra~E Todd.
\newblock Does matching overcome lalonde's critique of nonexperimental estimators?
\newblock \emph{Journal of Econometrics}, 125\penalty0 (1-2):\penalty0 305--353, 2005.

\bibitem[Stuart(2010)]{stuart2010matching}
Elizabeth~A Stuart.
\newblock Matching methods for causal inference: A review and a look forward.
\newblock \emph{Statistical Science}, 25\penalty0 (1):\penalty0 1--21, 2010.

\bibitem[Wang and Zubizarreta(2019)]{wang_minimal_2019}
Yixin Wang and Jose~R Zubizarreta.
\newblock Minimal dispersion approximately balancing weights: asymptotic properties and practical considerations.
\newblock \emph{Biometrika}, page asz050, October 2019.
\newblock ISSN 0006-3444, 1464-3510.
\newblock \doi{10.1093/biomet/asz050}.
\newblock URL \url{https://academic.oup.com/biomet/advance-article/doi/10.1093/biomet/asz050/5602475}.

\bibitem[White(1980)]{white1980heteroskedasticity}
Halbert White.
\newblock A heteroskedasticity-consistent covariance matrix estimator and a direct test for heteroskedasticity.
\newblock \emph{Econometrica: journal of the Econometric Society}, pages 817--838, 1980.

\bibitem[Zubizarreta(2015)]{zubizarreta2015stable}
Jos{\'e}~R Zubizarreta.
\newblock Stable weights that balance covariates for estimation with incomplete outcome data.
\newblock \emph{Journal of the American Statistical Association}, 110\penalty0 (511):\penalty0 910--922, 2015.

\end{thebibliography}

\newpage
\appendix

\section*{Appendix}

\section{Proof of Theorem~\ref{thm:clt} }
Note that
\begin{align*}
    \sqrt{n_T}\,(\hat \tau - B_n - \tau) &= \underline{\sqrt{n_T} ( \tau_{SATT} - \tau)} \\
    &\quad\text{Population error} \\
    &+ \underline{\sqrt{n_T} E_n } \\
    &\quad\text{measurement error}
\end{align*}
where $E_n = \left(\frac{1}{n_T} \sum_{t \in \mathcal{T}} \epsilon_t - \frac{1}{n_T} \sum_{j \in \mathcal{C}} w_j \epsilon_j \right)$

We focus on the population error and the measurement error separately.

First, we focus on the population error. By the Central Limit Theorem (CLT):
\begin{align*}
    &\sqrt{n_T} ( \tau_{SATT} - \tau)\\
    &= \sqrt{n_T} \left(\frac{1}{n_T} \sum_{t \in \mathcal{T}} \bigl(f_1(X_t) - f_0(X_t)\bigr) - \mathbb{E}_{X|Z=1}\bigl[f_1(X_i) - f_0(X_i) \mid Z_i = 1\bigr] \right)\\
    &\xrightarrow{d} \mathcal{N}(0, V_{P})
\end{align*}
where 
\begin{align}
    V_{P} = \mathbb{E}_{X|Z=1}\bigl[(f_1(X_i) - f_0(X_i) - \tau)^2 \mid Z_i = 1 \bigr]
\end{align}

Second, we focus on the measurement error part. We want to establish the asymptotic normality:
\begin{equation}
\frac{E_n}{\sqrt{V_E}} \xrightarrow{d} N(0,1)
\end{equation}

This is equivalent to showing that $\sqrt{n_T} E_n \xrightarrow{d} N(0,V_E)$ conditional on $\mathbf{X}, \mathbf{Z}$.

Let us denote $e_i = (Z_i - (1-Z_i)W_i)\epsilon_i$ as the weighted error contribution of the $i$-th unit, where $W_i = \sum_{t \in \mathcal{T}} w_{it}$ when $Z_i = 0$ and $W_i = 0$ otherwise. Then:
\begin{equation}
\frac{E_n}{\sqrt{V_E}} = \frac{1}{n_T\sqrt{V_E}} \sum_{i=1}^n e_i
\end{equation}

To establish the asymptotic normality, we verify the Lindeberg condition for triangular arrays, for all $\epsilon > 0$:
\begin{equation}
\frac{1}{s_n^2} \sum_{i=1}^n \mathbb{E}\left[T_{n,i}^2 \mathbf{1}(|T_{n,i}| > \epsilon s_n) \mid \mathbf{X}, \mathbf{Z}\right] \rightarrow 0 \text{ as } n\rightarrow \infty
\end{equation}
where
\begin{align*}
T_{n,i} &= \frac{e_i}{n_T\sqrt{V_E}} \\
\mathbb{E}[T_{n,i}^2 \mid \mathbf{X}, \mathbf{Z}] &= \frac{\mathbb{E}[e_i^2 \mid \mathbf{X}, \mathbf{Z}]}{n_T^2 V_E} \\
&= \frac{(Z_i - (1-Z_i)W_i)^2 \sigma_i^2}{n_T^2 V_E}
\end{align*}
and $s_n^2 = \sum_{i=1}^n \mathbb{E}[T_{n,i}^2 \mid \mathbf{X}, \mathbf{Z}] = 1$ \footnote{Recall that $V_E = \frac{1}{n_T^2} \left(\sum_{t \in \mathcal{T}} \sigma_{t}^2 + \sum_{j \in \mathcal{C}} (w_j)^2 \sigma_{j}^2 \right)$}

Hence, 
\begin{align}
    &\frac{1}{s_n^2} \sum_{i=1}^n \mathbb{E}\left[T_{n,i}^2 \mathbf{1}(|T_{n,i}| > \epsilon s_n) \mid \mathbf{X}, \mathbf{Z}\right] \\
    &= \sum_{i=1}^n \mathbb{E}\left[\left(\frac{e_i}{n_T\sqrt{V_E}}\right)^2 \mathbf{1}\left(\left|\frac{e_i}{n_T\sqrt{V_E}}\right| > \epsilon\right) \;\middle|\; \mathbf{X}, \mathbf{Z}\right] \\
    &= \frac{1}{n_T^2 V_E} \sum_{i=1}^n \mathbb{E}\left[e_i^2 \cdot \mathbf{1}(|e_i| > \epsilon n_T \sqrt{V_E}) \;\middle|\; \mathbf{X}, \mathbf{Z}\right] \label{eq:last-lindeberg}
\end{align}

Focusing on the $i$-th summand and applying Hölder's inequality with conjugate exponents $\frac{2+\delta}{2}$ and $\frac{2+\delta}{\delta}$:
\begin{align*}
    &\mathbb{E}\left[e_i^2 \cdot \mathbf{1}(|e_i| > \epsilon n_T \sqrt{V_E}) \;\middle|\; \mathbf{X}, \mathbf{Z}\right] \\
    &\leq \mathbb{E}\left[|e_i|^{2+\delta} \;\middle|\; \mathbf{X}, \mathbf{Z}\right]^{\frac{2}{2+\delta}} \cdot \mathbb{E}\left[\mathbf{1}(|e_i| > \epsilon n_T \sqrt{V_E}) \;\middle|\; \mathbf{X}, \mathbf{Z}\right]^{\frac{\delta}{2+\delta}}\\
    &= \mathbb{E}\left[|e_i|^{2+\delta} \;\middle|\; \mathbf{X}, \mathbf{Z}\right]^{\frac{2}{2+\delta}} \cdot \mathbb{P}\left[|e_i| > \epsilon n_T \sqrt{V_E} \;\middle|\; \mathbf{X}, \mathbf{Z}\right]^{\frac{\delta}{2+\delta}}\\
    &\leq \mathbb{E} \left[ |e_i|^{2+\delta} \,\big|\, \mathbf{X}, \mathbf{Z} \right]^{\frac{2}{2+\delta}}  \cdot \left(\frac{\mathbb{E} \left[ e_i^2 \,\big|\, \mathbf{X}, \mathbf{Z} \right]}{\epsilon^2 n_T^2 V_E}\right)^{\frac{\delta}{2 + \delta}}  \quad \text{by Markov's inequality} \\
    &= (Z_i - (1-Z_i) W_i)^{2 + \frac{2\delta}{2+\delta}} \cdot \frac{\mathbb{E}\left[|\epsilon_i|^{2+\delta} \;\middle|\; \mathbf{X}, \mathbf{Z}\right]^{\frac{2}{2+\delta}} \cdot \sigma_i^{\frac{2\delta}{2+\delta}}}{\epsilon^{\frac{2\delta}{2+\delta}} \cdot (n_T^2 V_E)^{\frac{\delta}{2+\delta}}} \\
    &\leq (Z_i - (1-Z_i) W_i)^{2 + \frac{2\delta}{2+\delta}} \cdot \frac{C^{\frac{2}{2+\delta}} \cdot \sigma_{max}^{\frac{2\delta}{2+\delta}}}{\epsilon^{\frac{2\delta}{2+\delta}} \cdot ( \sum_{i=1}^2 (Z_i - (1-Z_i) W_i)^2 \sigma_{min}^2 )^{\frac{\delta}{2+\delta}}}
\end{align*}

In the last step, we use the bounds from our assumptions: $\sigma_{min}^2 \leq \sigma_i^2 \leq \sigma_{max}^2$ from Assumption~\ref{assum:regular-variance} (Regular Variance) through the boundedness condition in Definition~\ref{def:regular-variance}, and the bound $\mathbb{E}\bigl[\bigl|\epsilon_i\bigr|^{\,2+\delta} \;\big|\; X_i = x \bigr] \leq C$ through the higher-order moment bound condition. 

Hence, the Lindeberg condition in Equation~\ref{eq:last-lindeberg} is upper bounded by
\begin{align}
   & \frac{1}{n_T^2 V_E} \sum_{i=1}^n \left[ (Z_i - (1-Z_i) W_i)^{2 + \frac{2\delta}{2+\delta}} \cdot \frac{C^{\frac{2}{2+\delta}} \cdot \sigma_{\max}^{\frac{2\delta}{2+\delta}}}{\epsilon^{\frac{2\delta}{2+\delta}} \cdot \left( \sum_{i=1}^n (Z_i - (1-Z_i) W_i)^2 \sigma_{\min}^2 \right)^{\frac{\delta}{2+\delta}}}\right] \\
\leq& \frac{C^{\frac{2}{2+\delta}} \cdot \sigma_{\max}^{\frac{2\delta}{2+\delta}}}{\epsilon^{\frac{2\delta}{2+\delta}} \cdot \sigma_{\min}^2 \cdot \left( \sum_{i=1}^n (Z_i - (1-Z_i) W_i)^2 \right)^{\frac{\delta}{2+\delta}}} \sum_{i=1}^n \left[ (Z_i - (1-Z_i) W_i)^{2 + \frac{2\delta}{2+\delta}}\right] \\
=& \frac{C^{\frac{2}{2+\delta}} \cdot \sigma_{\max}^{\frac{2\delta}{2+\delta}}}{\epsilon^{\frac{2\delta}{2+\delta}} \sigma_{\min}^{\frac{4 + 6\delta}{2+\delta}}} \cdot \frac{\sum_{i=1}^n \left[ (Z_i - (1-Z_i) W_i)^{2 + \frac{2\delta}{2+\delta}}\right]}{\left( \sum_{i=1}^n (Z_i - (1-Z_i) W_i)^2 \right)^{\frac{2 + 3\delta}{2+\delta}}} \\
=& \left(\frac{1}{n}\right)^{\frac{\delta}{2 + \delta}} \frac{C^{\frac{2}{2+\delta}} \cdot \sigma_{\max}^{\frac{2\delta}{2+\delta}}}{\epsilon^{\frac{2\delta}{2+\delta}} \sigma_{\min}^{\frac{4 + 6\delta}{2+\delta}}} \cdot \frac{\frac{1}{n} \sum_{i=1}^n \left[ (Z_i - (1-Z_i) W_i)^{2 + \frac{2\delta}{2+\delta}}\right]}{\left(\frac{1}{n} \sum_{i=1}^n (Z_i - (1-Z_i) W_i)^2 \right)^{\frac{2 + 3\delta}{2+\delta}}} \label{eq:final-lindeberg}
\end{align}

The term $\frac{\frac{1}{n} \sum_{i=1}^n \left[ (Z_i - (1-Z_i) W_i)^{2 + \frac{2\delta}{2+\delta}}\right]}{\left(\frac{1}{n} \sum_{i=1}^n (Z_i - (1-Z_i) W_i)^2 \right)^{\frac{2 + 3\delta}{2+\delta}}}$ is bounded in probability by Markov's inequality and because all moments of $W_i$ are bounded according to Lemma 3(i) in \cite{abadie2006large}. Therefore, Equation~\ref{eq:final-lindeberg} converges to 0 as $n \to \infty$ since $\left(\frac{1}{n}\right)^{\frac{\delta}{2 + \delta}} \to 0$. Thus, the Lindeberg condition is satisfied, establishing asymptotic normality.

Finally, $\sqrt{n_T} P_n =  \sqrt{n_T}(\tau_{SATT}-\tau) $ and $\sqrt{n_T}E_n$ are asymptotically independent, as the central limit theorem for $\sqrt{n_T}E_n$ holds conditional on the covariates $\mathbf{X}$ and treatment assignment $\mathbf{Z}$. Since both terms converge to normal distributions and given that $V_E$ is bounded and bounded away from zero, while $V_P$ remains bounded by the properties of the treatment effect function, we can conclude that 
\begin{equation}
\frac{\sqrt{n_T}(\hat{\tau} - B_n - \tau)}{\sqrt{V_E + V_P}} \xrightarrow{d} N(0,1)
\end{equation}
This establishes the asymptotic normality of our estimator after accounting for the bias term.

\section{Proof of Lemma~\ref{lemma:consistency-homogeneous}}
\label{sec:proof-lemma:consistency-homogeneous}
\begin{proof}
Let us decompose the difference between our variance estimator and the true average variance:
\begin{align*}
    S^2 - \frac{1}{n_T} \sum_{t=1}^{n_T} \sigma_t^2 
    &= \frac{1}{N_C} \sum_{t \in \mathcal{T}} |\mathcal{C}_t| s_t^2 - \frac{1}{n_T} \sum_{t=1}^{n_T} \sigma_t^2 \\
    &= \sum_{t \in \mathcal{T}} u_t s_t^2 - \frac{1}{n_T} \sum_{t \in \mathcal{T}} \sigma_t^2 \\
    &= \sum_{t \in \mathcal{T}} (u_t s_t^2 - \frac{1}{n_T} \sigma_t^2) \\
    &= \sum_{t \in \mathcal{T}} \underbrace{(u_t s_t^2 - u_t \sigma_t^2)}_{\text{Term A}} + \sum_{t \in \mathcal{T}} \underbrace{(u_t \sigma_t^2 - \frac{1}{n_T} \sigma_t^2)}_{\text{Term B}}
\end{align*}
where $u_t = \frac{|\mathcal{C}_t|}{N_C}$ represents the weight of cluster $t$ in the pooled estimator. Note that 
\begin{equation}
\label{eq:N_C}
N_C = \sum_{t \in \mathcal{T}} |\mathcal{C}_t|    
\end{equation}

is the total number of matches\footnote{If a control unit is matched to multiple treated units, it contributes to $N_C$ multiple times. For example, if a control unit is matched to three treated units, it adds 3 to $N_C$ rather than 1.}.

We first analyze Term A, which measures the difference between the estimated and true variance within each cluster. For a fixed treatment $t$, for each individual matched control $j$ in $\mathcal{C}_t$, we focus on the summand in $s_t^2 = \frac{1}{|\mathcal{C}_t| - 1} \sum_{j \in \mathcal{C}_t} (Y_{j} - \bar{Y}_t)^2$ (introduced in Equation~\ref{eq:s^2_t}) and expand the squared deviation:
\begin{align*}
(Y_j - \bar{Y}_t)^2 &= \left(f_0(X_j) - \frac{1}{|\mathcal{C}_t|} \sum_{k \in \mathcal{C}_t} f_0(X_k) + \epsilon_j - \frac{1}{|\mathcal{C}_t|} \sum_{k \in \mathcal{C}_t} \epsilon_k \right)^2 \\
&= \left(f_0(X_j) - \frac{1}{|\mathcal{C}_t|} \sum_{k \in \mathcal{C}_t} f_0(X_k)\right)^2 \\
&+ 2\left(f_0(X_j) - \frac{1}{|\mathcal{C}_t|} \sum_{k \in \mathcal{C}_t} f_0(X_k)\right)\left(\epsilon_j - \frac{1}{|\mathcal{C}_t|} \sum_{k \in \mathcal{C}_t} \epsilon_k\right) \\
&+ \left(\epsilon_j - \frac{1}{|\mathcal{C}_t|} \sum_{k \in \mathcal{C}_t} \epsilon_k\right)^2
\end{align*}

Therefore, the difference between the sample variance and the true variance can be written as:
\begin{align*}
s_t^2 - \sigma_t^2 &= \frac{1}{|\mathcal{C}_t|-1} \sum_{j \in \mathcal{C}_t} (Y_j - \bar{Y}_t)^2 - \sigma_t^2 \\
&= \underbrace{\left(\frac{1}{|\mathcal{C}_t|} \sum_{j \in \mathcal{C}_t} \epsilon_j^2 - \sigma_t^2\right)}_{\text{Sampling error}} \\
&+ \underbrace{\frac{1}{|\mathcal{C}_t|-1} \sum_{j \in \mathcal{C}_t}\left[-2\epsilon_j\left(\frac{1}{|\mathcal{C}_t|} \sum_{\substack{k \in \mathcal{C}_t \\ k \neq j}} \epsilon_k\right)\right]}_{\text{Cross-product of errors}} \\
&+ \underbrace{\frac{1}{|\mathcal{C}_t|-1} \sum_{j \in \mathcal{C}_t}\left[2\left(f_0(X_j) - \frac{1}{|\mathcal{C}_t|} \sum_{k \in \mathcal{C}_t} f_0(X_k)\right)\left(\epsilon_j - \frac{1}{|\mathcal{C}_t|} \sum_{k \in \mathcal{C}_t} \epsilon_k\right)\right]}_{\text{Interaction between function and errors}} \\
&+ \underbrace{\frac{1}{|\mathcal{C}_t|-1} \sum_{j \in \mathcal{C}_t}\left[\left(f_0(X_j) - \frac{1}{|\mathcal{C}_t|} \sum_{k \in \mathcal{C}_t} f_0(X_k)\right)^2\right]}_{\text{Systematic differences within cluster}}
\end{align*}
\end{proof}

Now, Term A becomes the following decomposition:
\begin{align*}
\text{Term A} &= \sum_{t=1}^{n_T} (u_t s_t^2 - u_t \sigma_t^2) \quad \text{where } u_t = \frac{|\mathcal{C}_t|}{\sum_{t=1}^{n_T} |\mathcal{C}_t|} \\
&= \underbrace{\sum_{t=1}^{n_T} \frac{u_t}{|\mathcal{C}_t|} \sum_{j \in \mathcal{C}_t} (\varepsilon_j^2 - \sigma_t^2)}_{\text{Sampling error}} \\
&+ \underbrace{\sum_{t=1}^{n_T} \frac{u_t}{|\mathcal{C}_t|} \sum_{j \in \mathcal{C}_t} \left[-2\varepsilon_j \left(\frac{1}{|\mathcal{C}_t|} \sum_{\substack{k \in \mathcal{C}_t \\ k \neq j}} \varepsilon_k \right) \right]}_{\text{Cross-product of errors}} \\
&+ \underbrace{\sum_{t=1}^{n_T} \frac{u_t}{|\mathcal{C}_t| - 1} \sum_{j \in \mathcal{C}_t} \left[-2 \left(f_0(X_j) - \frac{1}{|\mathcal{C}_t|} \sum_{k \in \mathcal{C}_t} f_0(X_k) \right) \left(\varepsilon_j - \frac{1}{|\mathcal{C}_t|} \sum_{k \in \mathcal{C}_t} \varepsilon_k \right) \right]}_{\text{Interaction between function and errors}} \\
&+ \underbrace{\sum_{t=1}^{n_T} \frac{u_t}{|\mathcal{C}_t| - 1} \sum_{j \in \mathcal{C}_t} \left[ \left(f_0(X_j) - \frac{1}{|\mathcal{C}_t|} \sum_{k \in \mathcal{C}_t} f_0(X_k) \right)^2 \right]}_{\text{Systematic differences within cluster}}
\end{align*}

Let's focus on the first component of Term A, the sampling error: 
\begin{align*}
& \sum_{t=1}^{n_1} \frac{u_t}{\left|C_t\right|} \sum_{j \in C t}\left(\varepsilon_j^2-\sigma_t^2\right) \\ 
= & \sum_{c=1}^{n_c} \sum_{t \in T_c} \frac{1}{\sum_{c=1}^{n_c} K(c) }\left(\varepsilon_c^2-\sigma_t^2\right) \\ 
= & \sum_{c=1}^{n_c} \sum_{t \in T_c} \frac{1}{\sum_{c=1}^{n_c}  K(c)}\left(\varepsilon_c^2-\sigma_c^2+\sigma_c^2-\sigma_t^2\right) \\ = & \frac{1}{\sum_{c=1}^{n_c} K(c)} \sum_{c=1}^{n_c} K(c)\left(\varepsilon_c^2-\sigma_c^2\right)+\frac{1}{\sum_{c=1}^{n_c} K(c)} \sum_{c=1}^{n_c} \sum_{t \in T_c}\left(\sigma_c^2-\sigma_t^2\right)
\end{align*}

where $K(c)$ represents the number of times control unit $c$ is used across all matches. Note that $ \sum_{c=1}^{n_C} K(c) = \sum_{t=1}^{n_T} |\mathcal{C}_t| = N_C$  is the total number of matches (Equation~\ref{eq:N_C}).

For the first term, we show that it converges to zero in probability using the approach from \cite{abadie2006large}. We establish this through a second moment condition and Chebyshev's inequality. The main proof relies on a Lemma:

\begin{lemma}[Bounded variance of squared error deviations]
\label{lemma:bounded-variance}
Under Assumption~\ref{assum:regular-variance} (Regular Variance), there exists a constant $C < \infty$ such that
$$\text{Var}(\varepsilon_c^2 - \sigma_c^2) \leq C$$
for all control units $c$.
\end{lemma}

Proof: See Section~\ref{sec:proof-lemma-bounded-variance}

We now prove that 
$$\frac{1}{\sum_{c=1}^{n_c} K(c)} \sum_{c=1}^{n_c} K(c)(\varepsilon_c^2-\sigma_c^2) \xrightarrow{p} 0$$

Let $S_n = \frac{1}{\sum_{c=1}^{n_c} K(c)} \sum_{c=1}^{n_c} K(c)(\varepsilon_c^2-\sigma_c^2)$.

\textbf{Step 1: Show the second moment condition}

We need to show that 
$$\frac{1}{(\sum_{c=1}^{n_c} K(c))^2} \sum_{c=1}^{n_c} K(c)^2 \xrightarrow{p} 0$$

We can rewrite this as:
$$\frac{1}{(\sum_{c=1}^{n_c} K(c))^2} \sum_{c=1}^{n_c} K(c)^2 = \frac{1}{(\sum_{c=1}^{n_c} K(c)/n_c)^2} \cdot \frac{1}{n_c^2} \sum_{c=1}^{n_c} K(c)^2$$

By Lemma 3 of \cite{abadie2006large}, all moments of $K(c)$ are finite. By the Law of Large Numbers:
$$\frac{\sum_{c=1}^{n_c} K(c)}{n_c} \xrightarrow{p} E[K(c)] \geq 1$$

where the lower bound follows since each control unit is matched at least once when used. Therefore:
$$\frac{1}{(\sum_{c=1}^{n_c} K(c)/n_c)^2} \xrightarrow{p} \frac{1}{E[K(c)]^2} < \infty$$

is bounded in probability.

For the second term, by the Law of Large Numbers:
$$\frac{1}{n_c} \sum_{c=1}^{n_c} K(c)^2 \xrightarrow{p} E[K(c)^2] < \infty$$

is bounded in probability by LLN. Therefore:
$$\frac{1}{n_c^2} \sum_{c=1}^{n_c} K(c)^2 = \frac{1}{n_c} \cdot \frac{1}{n_c} \sum_{c=1}^{n_c} K(c)^2 \xrightarrow{p} 0$$

since $\frac{1}{n_c} \to 0$ and $\frac{1}{n_c} \sum_{c=1}^{n_c} K(c)^2$ is bounded in probability.

Combining these results:
$$\frac{1}{(\sum_{c=1}^{n_c} K(c))^2} \sum_{c=1}^{n_c} K(c)^2 \xrightarrow{p} 0$$

\textbf{Step 2: Apply Chebyshev's inequality}

\textbf{Step 2a: Mean is zero}
$E[S_n] = \frac{1}{N_C} \sum_{c=1}^{n_c} K(c) \cdot E[\varepsilon_c^2-\sigma_c^2] = 0$

since $E[\varepsilon_c^2-\sigma_c^2] = E[\varepsilon_c^2] - \sigma_c^2 = 0$ by definition.

\textbf{Step 2b: Variance calculation}
Since the $\varepsilon_c$ are independent across control units:
$\text{Var}(S_n) = \text{Var}\left(\frac{1}{N_C} \sum_{c=1}^{n_c} K(c)(\varepsilon_c^2-\sigma_c^2)\right) = \frac{1}{N_C^2} \sum_{c=1}^{n_c} K(c)^2 \cdot \text{Var}(\varepsilon_c^2-\sigma_c^2)$

\textbf{Step 2c: Bound the variance}
By Lemma~\ref{lemma:bounded-variance}, $\text{Var}(\varepsilon_c^2-\sigma_c^2) \leq C$ for some constant $C$. Therefore:
$\text{Var}(S_n) \leq \frac{C}{N_C^2} \sum_{c=1}^{n_c} K(c)^2$

\textbf{Step 2d: Apply Step 1 result}
From Step 1, $\frac{1}{N_C^2} \sum_{c=1}^{n_c} K(c)^2 \to 0$, so:
$\text{Var}(S_n) \to 0$

\textbf{Step 2e: Chebyshev's inequality}
For any $\epsilon > 0$:
$P(|S_n| > \epsilon) \leq \frac{\text{Var}(S_n)}{\epsilon^2} \to 0$

Therefore:
$$\frac{1}{\sum_{c=1}^{n_c} K(c)} \sum_{c=1}^{n_c} K(c)(\varepsilon_c^2-\sigma_c^2) \xrightarrow{p} 0$$

For the second term:
\begin{align}
\frac{1}{\sum_{c=1}^{n_c} K(c)} \sum_{c=1}^{n_c} \sum_{t \in T_c}(\sigma_c^2-\sigma_t^2) &= \frac{1}{\sum_{c=1}^{n_c} K(c)} \sum_{c=1}^{n_c} K(c)(\sigma_c^2-\bar{\sigma}_c^2)
\end{align}

where $\bar{\sigma}_c^2 = \frac{1}{K(c)} \sum_{t \in T_c} \sigma_t^2$ is the average variance of the treated units matched to control unit $c$, and $K(c) = |T_c|$ represents the number of treated units to which control unit $c$ is matched.

We can bound this term as follows:
\begin{align}
\left|\frac{1}{\sum_{c=1}^{n_c} K(c)} \sum_{c=1}^{n_c} K(c)(\sigma_c^2-\bar{\sigma}_c^2)\right| &\leq \frac{1}{\sum_{c=1}^{n_c} K(c)} \sum_{c=1}^{n_c} K(c) \cdot \max_{c=1,\ldots,n_c}|\sigma_c^2-\bar{\sigma}_c^2|\\
&= \max_{c=1,\ldots,n_c}|\sigma_c^2-\bar{\sigma}_c^2| \xrightarrow{\text{a.s.}} 0 \text{ as } n_c, n_T \rightarrow \infty
\end{align}

where the last convergence follows from Lemma~\ref{lemma:uniform-variance-convergence}, which establishes the uniform convergence of variance differences across all control units.

For the second component of Term A (cross-product of errors):
\begin{align*}
& \sum_{t=1}^{n_T} \frac{u_t}{|\mathcal{C}_t|} \sum_{j \in \mathcal{C}_t}\left[-2 \varepsilon_j\left(\frac{1}{|\mathcal{C}_t|} \sum_{\substack{k \in \mathcal{C}_t \\ k \neq j}} \varepsilon_k\right)\right] \\
&= \sum_{t=1}^{n_T} \frac{u_t}{|\mathcal{C}_t|} \sum_{j \in \mathcal{C}_t}\left[-2 \varepsilon_j \frac{1}{|\mathcal{C}_t|} \sum_{\substack{k \in \mathcal{C}_t \\ k \neq j}} \varepsilon_k\right] \\
&= \sum_{t=1}^{n_T} \frac{1}{\sum_{t=1}^{n_T}|\mathcal{C}_t|} \frac{1}{|\mathcal{C}_t|} \sum_{\substack{j,k \in \mathcal{C}_t \\ j \neq k}}\left(-4 \varepsilon_j \varepsilon_k\right) \\
&\leq \frac{1}{\sum_{t=1}^{n_T}|\mathcal{C}_t|} \sum_{\substack{j,k \in \mathcal{C} \\ j \neq k}}-4 \cdot \frac{K(j,k)}{2} \varepsilon_j \varepsilon_k \\
&\leq \frac{1}{\sum_{t=1}^{n_T}|\mathcal{C}_t|} \sum_{\substack{j,k \in \mathcal{C} \\ j \neq k}}-2 \cdot K(j,k) \varepsilon_j \varepsilon_k
\end{align*}

where $K(j,k)$ represents the number of times control units $j$ and $k$ appear together in the same matched cluster. Since $|\mathcal{C}_t| \geq 2$ for all clusters (as we exclude singleton clusters), we have $\frac{1}{|\mathcal{C}_t|} \leq \frac{1}{2}$, which gives us the inequality in the last step. 

To establish that this term converges to zero in probability, we apply a similar two-step proof argument as above. The key observation is that $E[\varepsilon_j \varepsilon_k] = 0$ by independence, and under our matching schemes, $K(j,k)$ is asymptotically sparse—the probability that two specific units $j$ and $k$ are repeatedly matched together diminishes as $n_T$ increases. Note that $K(j,k) \leq \min\{K(j), K(k)\}$ since two units can appear together at most as many times as the less frequently used unit appears. Using the same second moment condition and Chebyshev's inequality approach, we can show that this cross-product term converges to zero in probability.

For the third component of Term A (interaction between function values and errors):
\begin{align*}
(A3) &= \sum_{t=1}^{n_T} \frac{u_t}{|\mathcal{C}_t|-1} \sum_{j \in \mathcal{C}_t}\left[-2\left(f_0(X_j)-\frac{1}{|\mathcal{C}_t|} \sum_{k \in \mathcal{C}_t} f_0(X_k)\right)\left(\varepsilon_j-\frac{1}{|\mathcal{C}_t|} \sum_{k \in \mathcal{C}_t} \varepsilon_k\right)\right] \\
\end{align*}

By the Mean Value Theorem and Assumption~\ref{assum:derivative-control}, we can bound the first factor:
\begin{align*}
\left|f_0(X_j)-\frac{1}{|\mathcal{C}_t|} \sum_{k \in \mathcal{C}_t} f_0(X_k)\right| &\leq \max_{k \in \mathcal{C}_t} |f_0(X_j) - f_0(X_k)| \\
&\leq \sup_{x \in \mathcal{C}_t} |f'_0(x)| \cdot \max_{j,k \in \mathcal{C}_t} \|X_j - X_k\| \\
&\leq \sup_{x \in \mathcal{C}_t} |f'_0(x)| \cdot r(\mathcal{C}_t)
\end{align*}

Therefore:
\begin{align*}
|(A3)| &\leq \sum_{t=1}^{n_T} \frac{u_t}{|\mathcal{C}_t|-1} \sum_{j \in \mathcal{C}_t} 2 \cdot \sup_{x \in \mathcal{C}_t} |f'_0(x)| \cdot r(\mathcal{C}_t) \cdot \left|\varepsilon_j-\frac{1}{|\mathcal{C}_t|} \sum_{k \in \mathcal{C}_t} \varepsilon_k\right| \\
&\leq 2 \cdot \sup_{t \in \mathcal{T}} \left[\sup_{x \in \mathcal{C}_t} |f'_0(x)| \cdot r(\mathcal{C}_t)\right] \cdot \sum_{t=1}^{n_T} \frac{u_t}{|\mathcal{C}_t|-1} \sum_{j \in \mathcal{C}_t} \left|\varepsilon_j-\frac{1}{|\mathcal{C}_t|} \sum_{k \in \mathcal{C}_t} \varepsilon_k\right| \\
&= 2 \cdot \sup_{t \in \mathcal{T}} \left[\sup_{x \in \mathcal{C}_t} |f'_0(x)| \cdot r(\mathcal{C}_t)\right] \cdot \frac{1}{\sum_{t=1}^{n_T}|\mathcal{C}_t|} \sum_{c=1}^{n_C} K(c) \left|\varepsilon_c - \frac{1}{|\mathcal{C}_t|} \sum_{k \in \mathcal{C}_t} \varepsilon_k\right|
\end{align*}

By Assumption~\ref{assum:derivative-control}, the first term $\sup_{t \in \mathcal{T}} [\sup_{x \in \mathcal{C}_t} |f'_0(x)| \cdot r(\mathcal{C}_t)]$ is bounded by a constant $M$. 

For the remaining term, we can apply a similar Hölder's inequality argument as developed for the sampling error term earlier. The structure involves products of $K(c)$ with error differences, which have the same statistical properties (independence, mean zero) as the $\varepsilon_c^2-\sigma_c^2$ terms analyzed above. Following the same steps—applying Hölder's inequality with conjugate exponents, leveraging Lemma 3 of \cite{abadie2006large} for the moments of $K(c)$, and using the higher-order moment bounds from Assumption~\ref{assum:regular-variance} (Regular Variance) through Definition~\ref{def:regular-variance}—we can establish that this term converges to zero in probability as $n_T \to \infty$.

Specifically, the weighted error differences satisfy:
\begin{align*}
\frac{1}{\sum_{t=1}^{n_T}|\mathcal{C}_t|} \sum_{c=1}^{n_C} K(c) \left|\varepsilon_c - \frac{1}{|\mathcal{C}_t|} \sum_{k \in \mathcal{C}_t} \varepsilon_k\right| \xrightarrow{p} 0
\end{align*}

as $n_T \to \infty$, by the convergence properties established in our analysis of the sampling error term. Therefore, $(A3) \xrightarrow{p} 0$ as $n_T \to \infty$.

For the fourth and final component of Term A (systematic differences within cluster):
\begin{align*}
(A4) &= \sum_{t=1}^{n_T} \frac{u_t}{|\mathcal{C}_t|-1} \sum_{j \in \mathcal{C}_t} \left[\left(f_0(X_j) - \frac{1}{|\mathcal{C}_t|} \sum_{k \in \mathcal{C}_t} f_0(X_k) \right)^2 \right] \\
\end{align*}

Similar to our analysis of term (A3), we can apply the Mean Value Theorem to bound each squared difference:
\begin{align*}
\left(f_0(X_j) - \frac{1}{|\mathcal{C}_t|} \sum_{k \in \mathcal{C}_t} f_0(X_k) \right)^2 &\leq \left(\max_{k \in \mathcal{C}_t} |f_0(X_j) - f_0(X_k)|\right)^2 \\
&\leq \left(\sup_{x \in \mathcal{C}_t} |f'_0(x)| \cdot \max_{j,k \in \mathcal{C}_t} \|X_j - X_k\|\right)^2 \\
&\leq \left(\sup_{x \in \mathcal{C}_t} |f'_0(x)| \cdot r(\mathcal{C}_t)\right)^2
\end{align*}

Thus:
\begin{align*}
|(A4)| &\leq \sum_{t=1}^{n_T} \frac{u_t}{|\mathcal{C}_t|-1} \sum_{j \in \mathcal{C}_t} \left(\sup_{x \in \mathcal{C}_t} |f'_0(x)| \cdot r(\mathcal{C}_t)\right)^2 \\
&= \sum_{t=1}^{n_T} \frac{u_t \cdot |\mathcal{C}_t|}{|\mathcal{C}_t|-1} \left(\sup_{x \in \mathcal{C}_t} |f'_0(x)| \cdot r(\mathcal{C}_t)\right)^2 \\
&\leq 2 \cdot \sum_{t=1}^{n_T} u_t \left(\sup_{x \in \mathcal{C}_t} |f'_0(x)| \cdot r(\mathcal{C}_t)\right)^2 \\
&\leq 2 \cdot \left(\sup_{t \in \mathcal{T}} \left[\sup_{x \in \mathcal{C}_t} |f'_0(x)| \cdot r(\mathcal{C}_t)\right]\right)^2
\end{align*}

By Assumption~\ref{assum:derivative-control}, $\sup_{x \in \mathcal{C}_t} |f'_0(x)| \cdot r(\mathcal{C}_t) \leq M$ for all $t$. Furthermore, by Assumption~\ref{assum:exponential-tail} (Shrinking Clusters), we have $\lim_{n \to \infty} \sup_{t \in \mathcal{T}} r(\mathcal{C}_t) = 0$. Therefore, even if $\sup_{x \in \mathcal{C}_t} |f'_0(x)|$ is unbounded as $n \to \infty$, their product with $r(\mathcal{C}_t)$ remains bounded by $M$, and the entire term $(A4) \xrightarrow{p} 0$ as $n_T \to \infty$.

For Term B, which involves the difference between the weighted and unweighted average of true variances:
\begin{align*}
\text{Term B} &= \sum_{t=1}^{n_T}\left(u_t \sigma_t^2 - \frac{1}{n_T} \sigma_t^2\right) \\
&= \sum_{t=1}^{n_T}\left(\frac{|\mathcal{C}_t|}{\sum_{t=1}^{n_T}|\mathcal{C}_t|} - \frac{1}{n_T}\right) \sigma_t^2
\end{align*}

By Assumption~\ref{assum:regular-variance} (Regular Variance) through the boundedness condition in Definition~\ref{def:regular-variance}, we know that $\sigma_{\min}^2 \leq \sigma_t^2 \leq \sigma_{\max}^2$ for all $t$. Therefore, by the triangle inequality:

\begin{align*}
|\text{Term B}| &\leq \sum_{t=1}^{n_T}\left|\frac{|\mathcal{C}_t|}{\sum_{t=1}^{n_T}|\mathcal{C}_t|} - \frac{1}{n_T}\right| \sigma_{\max}^2 \\
&= \sigma_{\max}^2 \sum_{t=1}^{n_T}\left|\frac{|\mathcal{C}_t|}{\sum_{t=1}^{n_T}|\mathcal{C}_t|} - \frac{1}{n_T}\right|
\end{align*}

Let $\bar{|\mathcal{C}|} = \frac{1}{n_T}\sum_{t=1}^{n_T}|\mathcal{C}_t|$ be the average cluster size. Then:

\begin{align*}
|\text{Term B}| &\leq \sigma_{\max}^2 \sum_{t=1}^{n_T}\left|\frac{|\mathcal{C}_t|}{n_T \cdot \bar{|\mathcal{C}|}} - \frac{1}{n_T}\right| \\
&= \frac{\sigma_{\max}^2}{n_T} \sum_{t=1}^{n_T}\left|\frac{|\mathcal{C}_t|}{\bar{|\mathcal{C}|}} - 1\right| \\
&= \frac{\sigma_{\max}^2}{n_T \cdot \bar{|\mathcal{C}|}} \sum_{t=1}^{n_T} \left||\mathcal{C}_t| - \bar{|\mathcal{C}|}\right|
\end{align*}

By the Law of Large Numbers, as $n_T \to \infty$, the variability in cluster sizes relative to their mean diminishes. Specifically, under our assumptions: The Shrinking Clusters Assumption~\ref{assum:exponential-tail} ensures that all clusters become increasingly homogeneous. For typical matching procedures like $M$-nearest neighbor matching, $|\mathcal{C}_t| = M$ for all $t$, making this term exactly zero.  For other matching procedures, the variance of $|\mathcal{C}_t|$ relative to $n_T$ approaches zero as $n_T$ increases.

This means that $\frac{1}{n_T} \sum_{t=1}^{n_T} \left||\mathcal{C}_t| - \bar{|\mathcal{C}|}\right| \xrightarrow{p} 0$ as $n_T \to \infty$. Therefore, Term B converges to zero in probability as $n_T \to \infty$, completing our proof that $S^2 \xrightarrow{p} \frac{1}{n_T} \sum_{t=1}^{n_T} \sigma_t^2$.

\section{Proof of Lemma~\ref{lemma:variance-equivalence}}
\label{sec:proof-lemma:variance-equivalence}
\begin{proof}
We begin by expanding the pooled variance estimator:

$$
\begin{aligned}
\hat{V}_{E, \text{lim}} & :=\left(\frac{1}{n_T} \sum_{t=1}^{n_T} \sigma_t^2\right)\left(\frac{1}{n_T}+\frac{1}{\operatorname{ESS}(\mathcal{C})}\right) \\
& =\frac{1}{n_T^2} \sum_{t=1}^{n_T} \sigma_t^2+\frac{1}{n_T} \sum_{t=1}^{n_T} \sigma_t^2\left(\frac{\sum_{j \in \mathcal{C}} w_j^2}{n_T^2}\right) \\
& =\frac{1}{n_T^2} \sum_{t=1}^{n_T} \sigma_t^2+\frac{1}{n_T^2} \sum_{j=n_T+1}^{n_T+n_C}\left[w_j^2 \cdot\left(\frac{1}{n_T} \sum_{t=1}^{n_T} \sigma_t^2\right)\right] \\
\end{aligned}
$$

The error variance $V_E$ can be similarly expanded:

$$
\begin{aligned}
V_E &=\frac{1}{n_T^2}\left(\sum_{t=1}^{n_T} \sigma_t^2+\sum_{j=n_T+1}^{n_T+n_C} w_j^2 \sigma_j^2\right)\\
&=\frac{1}{n_T^2} \sum_{t=1}^{n_T} \sigma_t^2+\frac{1}{n_T^2} \sum_{j=n_T+1}^{n_T+n_C} w_j^2 \sigma_j^2
\end{aligned}
$$

To establish asymptotic equivalence, we analyze the difference:

$$
\begin{aligned} 
\hat{V}_{E, \text{lim}}-V_E &= \frac{1}{n_T^2} \sum_{j=n_T+1}^{n_T+n_C}\left[w_j^2 \cdot\left(\frac{1}{n_T} \sum_{t=1}^{n_T} \sigma_t^2\right)\right]-\frac{1}{n_T^2} \sum_{j=n_T+1}^{n_T+n_C} w_j^2 \sigma_j^2 \\ 
&= \frac{1}{n_T^2} \sum_{j=n_T+1}^{n_T+n_C} w_j^2 \cdot\left[\left(\frac{1}{n_T} \sum_{t=1}^{n_T} \sigma_t^2\right)-\sigma_j^2\right] \\ 
\end{aligned}
$$

Using the definition of $\operatorname{ESS}(\mathcal{C})$, we can rewrite:

$$
\begin{aligned}
\hat{V}_{E, \text{lim}}-V_E &= \frac{1}{n_T} \frac{1}{\operatorname{ESS}(\mathcal{C})} \sum_{t=1}^{n_T} \sigma_t^2-\frac{1}{n_T^2} \sum_{j=n_T+1}^{n_T+n_C} w_j^2 \sigma_j^2
\end{aligned}
$$

We decompose this difference into two terms:

$$
\begin{aligned} 
\hat{V}_{E, \text{lim}}-V_E &= \underbrace{\left(\frac{1}{n_T} \frac{1}{\operatorname{ESS}(\mathcal{C})} \sum_{t=1}^{n_T} \sigma_t^2-\frac{1}{n_T} \frac{1}{\operatorname{ESS}(\mathcal{C})} \sum_{t=1}^{n_T} \overline{\sigma_t^2}\right)}_{(I)} \\ 
&+ \underbrace{\left(\frac{1}{n_T} \frac{1}{\operatorname{ESS}(\mathcal{C})} \sum_{t=1}^{n_T} \overline{\sigma_t^2}-\frac{1}{n_T^2} \sum_{j=n_T+1}^{n_T+n_C} w_j^2 \sigma_j^2\right)}_{(II)}
\end{aligned}
$$

where $\overline{\sigma_t^2} = \sum_{j \in \mathcal{C}_t} w_{jt} \sigma_j^2$ represents the weighted average of variances for control units matched to treated unit $t$.

For Term (I), we have:

$$
\begin{aligned}
(I)&=\frac{1}{n_T} \frac{1}{\operatorname{ESS}(\mathcal{C})} \sum_{t=1}^{n_T}\left(\sigma_t^2-\overline{\sigma_t^2}\right) \xrightarrow{p} 0
\end{aligned}
$$

This convergence follows from Assumption~\ref{assum:regular-variance} (Regular Variance) through the continuity condition in Definition~\ref{def:regular-variance}. As the matching quality improves under Assumption \ref{assum:exponential-tail} (Shrinking Clusters), the difference between $\sigma_t^2$ and $\overline{\sigma_t^2}$ diminishes because matched control units have variance values increasingly similar to their corresponding treated units.

For Term (II), we analyze:

$$
\begin{aligned}
\frac{1}{n_T} \frac{1}{\operatorname{ESS}(\mathcal{C})} \sum_{t=1}^{n_T} \overline{\sigma_t^2} 
&= \frac{1}{n_T} \frac{1}{\operatorname{ESS}(\mathcal{C})} \sum_{t=1}^{n_T} \sum_{j \in \mathcal{C}_t} w_{jt} \sigma_j^2 \\
&= \frac{1}{n_T} \frac{1}{\operatorname{ESS}(\mathcal{C})} \sum_{j=n_T+1}^{n_T+n_C} w_j \sigma_j^2 \quad \text{where} \; w_j=\sum_{t=1}^{n_T} w_{jt} \\
&= \frac{1}{n_T} \frac{\sum_{j \in \mathcal{C}} w_j^2}{n_T^2} \sum_{j=n_T+1}^{n_T+n_C} w_j \sigma_j^2 
\end{aligned}
$$

Substituting this into Term (II):

$$
\begin{aligned}
(II)&=\frac{1}{n_T^2} \sum_{j=n_T+1}^{n_T+n_C}\left(\frac{\sum_{j' \in \mathcal{C}} w_{j'}^2}{n_T} w_j-w_j^2\right) \sigma_j^2 \\
&=\frac{1}{n_T} \frac{1}{n_T} \sum_{j=n_T+1}^{n_T+n_C}\left(\frac{\sum_{j' \in \mathcal{C}} w_{j'}^2}{n_T} w_j-w_j^2\right) \sigma_j^2
\end{aligned}
$$

In the limiting case with K-NN matching using uniform weighting, each control unit $j$ receives weight $w_j = 1/K$ if matched to exactly one treated unit, and $w_j = 0$ if unmatched. 

For a matched control unit $j$ with $w_j = 1/K$:
\begin{itemize}
\item $\sum_{j' \in \mathcal{C}} w_{j'}^2 = n_T \cdot K \cdot (1/K)^2 = n_T/K$
\item The expression becomes: $\frac{n_T/K}{n_T} \cdot \frac{1}{K} - \left(\frac{1}{K}\right)^2 = \frac{1}{K^2} - \frac{1}{K^2} = 0$
\end{itemize}

For an unmatched control unit $j$ with $w_j = 0$, the expression trivially equals 0.

Therefore, $\left(\frac{\sum_{j' \in \mathcal{C}} w_{j'}^2}{n_T} w_j-w_j^2\right) = 0$ for all $j$, implying \textbf{Term (II) = 0} in the limiting case.

Since both Term (I) and Term (II) converge to zero in probability, we conclude:

$$\left|\hat{V}_{E, \text{lim}}-V_E\right| \xrightarrow{p} 0 \quad \text{as} \; n_T \rightarrow \infty$$

This establishes the variance-weighting equilibrium: when $V_{E,\lim} - V_E \xrightarrow{P} 0$, the weights $w_j$ and the effective sample size $\frac{1}{\text{ESS}(\mathcal{C})}$ perfectly balance each other, making the control variance contribution equivalent to a reweighted version of the treated variance.

\end{proof}

\section{Proof of Lemma~\ref{lemma:bounded-variance}}
\label{sec:proof-lemma-bounded-variance}

\begin{proof}
Since $\sigma_c^2 = E[\varepsilon_c^2 | X_c]$, we have $E[\varepsilon_c^2 - \sigma_c^2] = 0$. Therefore:
$$\text{Var}(\varepsilon_c^2 - \sigma_c^2) = E[(\varepsilon_c^2 - \sigma_c^2)^2]$$

Expanding the squared term:
$$(\varepsilon_c^2 - \sigma_c^2)^2 = \varepsilon_c^4 - 2\varepsilon_c^2\sigma_c^2 + (\sigma_c^2)^2$$

Taking expectation:
$$E[(\varepsilon_c^2 - \sigma_c^2)^2] = E[\varepsilon_c^4] - 2E[\varepsilon_c^2\sigma_c^2] + E[(\sigma_c^2)^2]$$

Since $\sigma_c^2 = E[\varepsilon_c^2 | X_c]$ is deterministic given $X_c$:
$$E[\varepsilon_c^2\sigma_c^2] = E[E[\varepsilon_c^2 | X_c] \cdot \sigma_c^2] = E[(\sigma_c^2)^2]$$

Therefore:
$$\text{Var}(\varepsilon_c^2 - \sigma_c^2) = E[\varepsilon_c^4] - E[(\sigma_c^2)^2]$$

By Definition~\ref{def:regular-variance}, condition 3 (Higher-order moment bound), with $\delta \geq 2$:
$$E[\varepsilon_c^4] = E[E[\varepsilon_c^4 | X_c]] = E[E[|\varepsilon_c|^{2+2} | X_c]] \leq E[C] = C$$

By Definition~\ref{def:regular-variance}, condition 2 (Boundedness):
$$E[(\sigma_c^2)^2] \geq 0$$

Therefore:
$$\text{Var}(\varepsilon_c^2 - \sigma_c^2) = E[\varepsilon_c^4] - E[(\sigma_c^2)^2] \leq C - 0 = C$$
\end{proof}

\section{Other useful Lemmas}

\subsection{Lemma and Proof of Shrinking Cluster Distance under Compact Support}

\begin{lemma}[Compact Support $\implies$ Vanishing Matching Discrepancy]\label{lemma:shrinking}
Let $\mathcal{X}\subset\mathbb{R}^d$ be the support of the covariates $X$, and assume $\mathcal{X}$ is compact. Assume further that the distribution of $X$ has a density $f_X(x)$ that is bounded above and below on $\mathcal{X}$ (i.e., $0 < f_{\min} \le f_X(x)\le f_{\max} < \infty$ for all $x\in\mathcal{X}$). Consider any matching procedure that pairs each observation with at least one other “nearest” neighbor (for example, one-to-$M$ nearest neighbor matching or radius matching with a fixed caliper). Then as the sample size $N\to\infty$, the maximum distance between any matched units goes to zero. In particular:
$$
\max_{i=1,\dots,N} \min_{j \neq i} \|X_i - X_j\| ~\xrightarrow{p}~ 0.
$$ 
That is, the distance between each observation and its closest match converges to zero in probability (and in fact, almost surely).
\end{lemma}

\begin{proof}[Proof Sketch]
This result is a direct consequence of the compactness of $\mathcal{X}$ and the bounded-positive density assumption. The argument follows the intuition of Lemma 1 in \citet{abadie2006large}.

Because $\mathcal{X}$ is compact, for any given radius $\varepsilon>0$ we can **cover $\mathcal{X}$ by finitely many small balls** (or other simple sets) of diameter at most $\varepsilon$.  Concretely, by the Heine–Borel covering theorem, there exists a finite collection of sets $\{B_1,\ldots,B_R\}$ such that $\mathcal{X}\subseteq \bigcup_{r=1}^R B_r$ and each $B_r$ has $\text{diam}(B_r) < \varepsilon$. For example, one can take $B_r$ to be balls (in $\|\cdot\|$ norm) of radius $\varepsilon/2$, or cubes of side-length $\varepsilon$, etc., partitioning the space. By construction, for any $x,x'$ in the same $B_r$, we have $\|x - x'\| < \varepsilon$.

Next, because $f_X(x)\ge f_{\min} > 0$ on $\mathcal{X}$, *every region of $\mathcal{X}$ has some probability mass*. In particular, each $B_r$ has $\Pr(X\in B_r) > 0$. When we draw $N$ i.i.d. observations $\{X_i\}_{i=1}^N$, the expected number of samples falling in $B_r$ is $N\Pr(X\in B_r)$, which grows linearly with $N$. By the law of large numbers, for large $N$ it is overwhelmingly likely that each $B_r$ contains at least one sample point (indeed, at least $\approx N\Pr(X\in B_r)$ points). More strongly, since $N\Pr(X\in B_r) \to\infty$, the probability that any given $B_r$ contains **fewer than 2 points** goes to zero as $N\to\infty$. In fact, one can apply a union bound or a Poisson approximation to show:
$$
P\Big(\exists~r:~ B_r \text{ contains 0 or 1 points}\Big) ~\to~ 0, \quad \text{as } N\to\infty.
$$ 
(Informally: with infinitely many draws, every subset $B_r$ will eventually have multiple points due to the density’s support.)

Now consider any observation $i$ and let $B_r$ be one of the covering sets that contains $X_i$. By the above argument, for large $N$ that $B_r$ will contain at least one **other** observation $j \ne i$. Thus, $i$ has at least one neighbor $j$ with $X_j \in B_r$ alongside $X_i$. By the diameter property of $B_r$, the distance between $i$ and this neighbor $j$ is bounded by $\|X_i - X_j\| < \varepsilon$. Since $i$ was arbitrary, we have shown that **for every $i$ there exists some match $j$ with $\|X_i - X_j\| < \varepsilon$** (with probability $\to 1$ as $N$ large).

Because $\varepsilon>0$ was arbitrary, it follows that the maximum matching distance in the sample is $<\varepsilon$ w.p.~$\to1$ for any $\varepsilon$. In probabilistic terms:
$$
\Pr\Big\{\max_{1\le i\le N}\min_{j\neq i}\|X_i-X_j\| < \varepsilon\Big\} ~\to~ 1 \quad \forall~\varepsilon>0,
$$ 
which is equivalent to $\max_i \min_{j}\|X_i-X_j\| \xrightarrow{p} 0$. (In fact, one can show almost sure convergence to 0 by invoking the Borel–Cantelli lemma, since the event that a given $B_r$ is empty eventually occurs at most finitely many times.)

This proves that the largest distance within any matched pair (or cluster) converges to zero as $N$ increases, under the stated assumptions. $\hfill\square$
\end{proof}

Remark: This lemma provides the theoretical justification for Assumption~\ref{assum:exponential-tail} in our paper. It confirms that under compact support (and overlap), **common matching methods produce asymptotically exact matches**. Notably, nearest-neighbor matching on $\mathbf{X}$ yields $\|\widehat X_i - X_i\| = O_p(N^{-1/d})$, so the discrepancy vanishes as $N\to\infty$. This is analogous to the overlap condition in propensity score matching, where a bounded propensity support guarantees treated units find control units with arbitrarily close propensity scores. Conversely, if covariate support were unbounded or the density went to zero in some region, one could not guarantee such shrinking distances – there would always be a chance of an isolated observation with no close neighbor (e.g. an outlier), resulting in a non-vanishing maximum discrepancy.

\subsection{Lemma for Proof of Theorem~\ref{lemma:consistency-homogeneous}}

\begin{lemma}[Uniform convergence of variances]
\label{lemma:uniform-variance-convergence}
Under Assumptions~\ref{assum:exponential-tail} (Shrinking Clusters) and~\ref{assum:regular-variance} (Regular Variance) through the continuity condition in Definition~\ref{def:regular-variance}, we have:
\begin{equation}
\max_{c=1,\ldots,n_c}|\sigma_c^2 - \bar{\sigma}_c^2| \xrightarrow{\text{a.s.}} 0 \quad \text{as } n_c, n_T \to \infty
\end{equation}
where $\sigma_c^2 = \sigma^2(X_c)$ and $\bar{\sigma}_c^2 = \frac{1}{K(c)}\sum_{t \in T_c}\sigma^2(X_t)$ with $T_c$ being the set of treated units matched to control unit $c$.
\end{lemma}

\begin{proof}
Let us establish a framework for proving the uniform convergence. For a given sample size $n = n_c + n_T$, define $d_n$ as the maximum matching distance such that $T_c = \{t \in \mathcal{T}: \|X_c - X_t\| \leq d_n\}$ for each control unit $c$. Under Assumption~\ref{assum:exponential-tail} (Shrinking Clusters), we have $d_n \xrightarrow{\text{a.s.}} 0$ as $n \to \infty$.

By Assumption~\ref{assum:regular-variance} (Regular Variance) through the continuity condition in Definition~\ref{def:regular-variance}, there exists a Lipschitz constant $L$ such that:
\begin{equation}
|\sigma^2(x) - \sigma^2(y)| \leq L \cdot \|x - y\|
\end{equation}
for all $x, y \in \mathcal{X}$.

For any control unit $c$, we have:
\begin{align*}
|\sigma_c^2 - \bar{\sigma}_c^2| &= \left|\sigma^2(X_c) - \frac{1}{K(c)}\sum_{t \in T_c}\sigma^2(X_t)\right| \\
&\leq \frac{1}{K(c)}\sum_{t \in T_c}|\sigma^2(X_c) - \sigma^2(X_t)| \quad \text{(by triangle inequality)} \\
&\leq \frac{1}{K(c)}\sum_{t \in T_c}L \cdot \|X_c - X_t\| \quad \text{(by Lipschitz condition)}
\end{align*}

Since all $t \in T_c$ satisfy $\|X_c - X_t\| \leq d_n$ by construction, we have:
\begin{align*}
|\sigma_c^2 - \bar{\sigma}_c^2| &\leq \frac{1}{K(c)}\sum_{t \in T_c}|\sigma^2(X_c) - \sigma^2(X_t)| \\
&\leq \frac{1}{K(c)}\sum_{t \in T_c}L \cdot \|X_c - X_t\| \\
&\leq \frac{L}{K(c)}\sum_{t \in T_c} d_n \\
&= L \cdot d_n
\end{align*}

This inequality holds uniformly for every control unit $c$, as the Lipschitz constant $L$ applies across all matches and $d_n$ represents the maximum matching distance. Therefore, the maximum deviation across all control units is bounded by:
\begin{align*}
\max_{c=1,\ldots,n_c}|\sigma_c^2 - \bar{\sigma}_c^2| \leq L \cdot d_n
\end{align*}

By Assumption~\ref{assum:exponential-tail}, $d_n \xrightarrow{\text{a.s.}} 0$ as $n \to \infty$. Since $L$ is a finite constant, we conclude:
\begin{equation}
\max_{c=1,\ldots,n_c}|\sigma_c^2 - \bar{\sigma}_c^2| \xrightarrow{\text{a.s.}} 0 \quad \text{as } n_c, n_T \to \infty
\end{equation}

This establishes the uniform convergence of variance differences across all control units.

This establishes the uniform convergence of $\sigma_c^2 - \bar{\sigma}_c^2$ across all control units simultaneously, not merely pointwise convergence for each fixed $c$.
\end{proof}

\section{Compare the Lipschitz Condition to that in the Existing Litarature}

In the existing literature, the function $f(x)$ is often assumed to be locally Lipschitz on any compact set $\mathcal{X} \subset \mathbb{R}$. This implies that for any compact set $\mathcal{X} = [a, b]$, there exists a constant $L_\mathcal{X} < \infty$ such that:
\[
|f(x) - f(y)| \leq L_\mathcal{X} |x - y|, \quad \forall x, y \in \mathcal{X}.
\]
For example, consider $f(x) = x^2$, where the derivative $f'(x) = 2x$. On $\mathcal{X} = [0, 100]$, the Lipschitz constant is:
\[
L_\mathcal{X} = 2 \cdot \max_{x \in \mathcal{X}} |x| = 200.
\]
This large constant makes the bound impractical in matching-based inference, where overly conservative bounds can restrict the formation of matched sets. 

In contrast, our Derivative Control condition improves on the Lipschitz assumption by explicitly tying the slope of $f(x)$ to the size of the matched set. Specifically, it requires:
\[
\sup_{x \in \mathcal{C}_t} \bigl|f'(x)\bigr| \cdot \mathrm{radius}(\mathcal{C}_t) \leq M,
\]
where:
\begin{itemize}
    \item $\mathcal{C}_t$ is the matched set for a given $t$,
    \item $\mathrm{radius}(\mathcal{C}_t)$ is the diameter of the matched set in $x$-space,
    \item $M$ is a universal constant independent of the matched set size.
\end{itemize}

This condition offers several practical advantages:
\begin{enumerate}
    \item Localized Control: Instead of requiring a single large Lipschitz constant $L_\mathcal{X}$ over a wide range, our condition focuses on smaller, localized matched sets.
    \item Adaptive Bounds: When the derivative $f'(x)$ is large, our condition naturally enforces smaller matched set radii to maintain practical bounds. For instance:
    \[
    \text{If } f'(x) = 100 \text{ (as for } x = 50 \text{), then } \mathrm{radius}(\mathcal{C}_t) \leq \frac{M}{100}.
    \]
    \item Real-World Applicability: In real-world matching problems, matched sets are typically small, and our condition aligns with this reality by providing sharper, more practical bounds than the overly conservative Lipschitz constant.
\end{enumerate}

To summarize, while the Lipschitz assumption is valid on compact sets, the associated constants $L_\mathcal{X}$ can become impractically large for functions like $f(x) = x^2$ over wide intervals. By explicitly accounting for both the derivative and the size of matched sets, our condition provides a more precise and practical framework for matching-based inference.

\section{Comparison with Theorem 1 of \cite{white1980heteroskedasticity}}
\label{sec:comparison-to-white-HC}

Our theorem, stated as Theorem~\ref{thm:consistency-total-variance}, differs from Theorem 1 of \cite{white1980heteroskedasticity} in several key aspects. While both results address consistency in variance estimation under heteroskedasticity, the differences lie in the frameworks, assumptions, and proof strategies.

\subsection{Parametric vs. Nonparametric Framework}

White's Theorem 1 is based on a regression model \(Y_i = X_i \beta_0 + \varepsilon_i\), where \(\varepsilon_i\) represents independent but non-identically distributed (i.n.i.d.) errors. The parametric form \(X_i \beta_0\) is central, and \(\beta_0\) is estimated via ordinary least squares (OLS). Heteroskedasticity arises through \(\mathrm{Var}(\varepsilon_i \mid X_i) = g(X_i)\), where \(g(X_i)\) is a known (possibly parametric) function. In contrast, our theorem relies on a nonparametric matching estimator for treatment effects, without assuming a parametric form for \(f(X_i)\). Matching is governed by hyperparameters like the number of neighbors or the maximum matching radius, but these are not estimated from the data in the regression sense. Heteroskedasticity arises through \(\sigma^2(X_i)\), where \(\sigma^2(\cdot)\) is a uniformly continuous function.

\subsection{White's Setup: Estimating \(\mathrm{Var}(\hat{\beta})\) vs. Cluster-Based Variance Estimation}

White's Theorem 1 focuses on the heteroskedasticity-consistent (HC) covariance matrix estimator for \(\hat{\beta}\). It defines the matrix
\[
\hat{V}_n = \frac{1}{n} \sum_{i=1}^n \hat\varepsilon_i^2 X_i^\prime X_i,
\quad \text{where} \quad \hat\varepsilon_i = Y_i - X_i \hat{\beta}.
\]
White proves \(\hat{V}_n \xrightarrow{\text{a.s.}} \bar{V}_n\), where \(\bar{V}_n\) is the asymptotic covariance matrix of the regressors. Our theorem, on the other hand, defines cluster-level residual variance estimators \(s_t^2\) for each treated unit \(t \in \mathcal{T}\), given its matched controls \(\mathcal{C}_t\). The overall variance estimator is
\[
S^2 = \frac{1}{n_T} \sum_{t=1}^{n_T} s_t^2,
\quad \text{where} \quad
s_t^2 = \frac{1}{|\mathcal{C}_t| - 1} \sum_{j \in \mathcal{C}_t} e_{tj}^2.
\]
We prove \(|S^2 - \frac{1}{n_T}\sum_{t=1}^{n_T}\sigma_t^2| \xrightarrow{\text{a.s.}} 0\), showing consistency for the average cluster variance.

\subsection{Homoskedasticity in Matched Clusters vs. General Heteroskedasticity}

White's Theorem 1 allows general heteroskedasticity: \(\mathrm{Var}(\varepsilon_i \mid X_i) = g(X_i)\), where \(g(\cdot)\) can vary arbitrarily across observations. Errors are independent but not identically distributed (i.n.i.d.). Our theorem also allows heteroskedasticity: \(\sigma^2(X_i)\) varies with \(X_i\). However, within each matched cluster \(\{t\} \cup \mathcal{C}_t\), we assume \(\sigma_j^2 \approx \sigma_t^2\) for \(j \in \mathcal{C}_t\), based on a uniform continuity (or Lipschitz) assumption on \(\sigma^2(\cdot)\).

\subsection{Proof Strategy and Key Assumptions}

White's proof strategy relies on expanding \(\hat{V}_n - \bar{V}_n\) and showing that
\[
\hat{V}_n - \bar{V}_n = \frac{1}{n}\sum_{i=1}^n \bigl(\hat\varepsilon_i^2 X_i^\prime X_i - E[\varepsilon_i^2 X_i^\prime X_i]\bigr) \xrightarrow{\text{a.s.}} 0.
\]
White uses assumptions on finite moments of \(\varepsilon_i\) and \(X_i\) (Assumptions 2–4 in \cite{white1980heteroskedasticity}) and uniform integrability conditions. Our proof, in contrast, relies on showing that for matched clusters \(\{t\} \cup \mathcal{C}_t\), the residual variance \(s_t^2\) converges to the true variance \(\sigma_t^2\). We leverage uniform continuity of \(\sigma^2(\cdot)\) to argue that \(\sigma_j^2 \to \sigma_t^2\) as \(\|X_{tj} - X_t\| \to 0\). We then apply a version of the Law of Large Numbers (LLN) for matched clusters.

\subsection{Summary of Differences}

The key differences between White's theorem and our theorem can be summarized as follows. First, White's theorem is regression-based, while our theorem is matching-based. Second, White assumes a parametric model \(Y_i = X_i \beta_0 + \varepsilon_i\), whereas our model assumes a nonparametric \(f_1(X)\), \(f_0(X)\). Third, White's focus is on a robust covariance estimator for \(\hat{\beta}\), while ours is on residual variance from matched clusters. Fourth, White allows fully general \(g(X_i)\), whereas our clusters assume approximate homoskedasticity (\(\sigma_j^2 \approx \sigma_t^2\)). Finally, White's framework has no matching hyperparameters, while ours depends on predefined criteria for matching (e.g., number of neighbors or radius).

\section{Otsu and Rai Variance Estimator}
\subsubsection{Debiasing Method}
A debiasing model estimates the conditional mean function $\mu(z, x) = E[Y \mid Z=z, X=x]$. It is used to offset the bias to achieve valid inference (see Section~\ref{sec:CLT} for discussion of the issue). The debiased estimator is defined as:
\begin{equation}
\Tilde{\tau}(w) = \frac{1}{n_T} \sum_{t \in \mathcal{T}} \left(Y_t - \hat \mu(0, X_t) - \sum_{j \in \mathcal{C}_t} w_{jt} ( Y_j - \hat\mu(0, X_j) ) \right)
\end{equation}

Additional implementation details include:
\begin{itemize}
\item \textbf{Model Choice}: Linear model
\item \textbf{Training Data}: Control data only
\item \textbf{Cross-fitting}: Implemented by dividing the control data into two halves
\end{itemize}

\subsubsection{Variance Estimators}

\paragraph{Bootstrap Variance Estimator.}
\begin{itemize}
    \item Step 1: Use data with $Z_i = 0$ to construct $\hat \mu(0, x) = \hat E[Y|Z=0, X=x]$.
    \item Step 2: Construct debiased estimate for each treated unit $t \in \mathcal{T}$:
    \begin{align*}
        \tilde \tau_t = (Y_t - \hat \mu(0, X_t)) - \sum_{j \in \mathcal{C}_t} w_{jt} (Y_j - \hat \mu(0, X_j))
    \end{align*}
    \item Step 3: Construct the debiased estimator: $\tilde \tau = \frac{1}{n_t} \sum_{t \in \mathcal{T}} \tilde \tau_t$ 
    \item Step 4: Construct the debiased residuals $R_t = \tilde \tau_t - \tilde \tau$ 
    \item Step 5: Perform Wild bootstrap on $\{R_t \}$ with special sampling weights
    \item Step 6: Construct confidence interval from bootstrap distribution
\end{itemize}

\paragraph{Pooled Variance Estimator.}
\begin{itemize}
    \item Step 1: Obtain the debiased estimator $\tilde \tau_t$
    \item Step 2: Estimate the variance using:
\begin{equation}
   \hat{V} = S^2 \left( \frac{1}{n_T} + \frac{1}{\text{ESS}(\mathcal{C})} \right)
\end{equation}
    where $S^2$ is a pooled variance estimator across clusters of treated and matched controls
    \item Step 3: Construct the 95\% confidence interval by $\tilde \tau \pm 1.96 * \sqrt{\hat V}$
\end{itemize}

\section{Other Simulation Results}
\label{sec:other-simulation-results}

\subsection{Additional Simulation Results of the \cite{che2024caliper} DGP}
\label{sec:additional-simulation-results}

This section provides supplementary simulation results that further validate our theoretical framework. We examine three key aspects: the accuracy of our $V_E$ component estimation, verification of asymptotic bias patterns, and the behavior of effective sample sizes across different overlap scenarios.

\begin{table}[ht]
\centering
\caption{Additional Simulation Results: Variance Components and Bias Analysis}
\label{tab:additional-simulation-results}
\begin{tabular}{lcccccc}
\toprule
\textbf{Degree} & \textbf{True} & \textbf{Est.} & \textbf{Coverage} & \textbf{Coverage w/o} & \textbf{Mean} & \textbf{Mean} \\
\textbf{of Overlap} & $SE_E$ & $SE_E$ & \textbf{Rate} & \textbf{Bias Corr.} & \textbf{ESS}$_C$ & $V/n_T$ \\
\midrule
Very Low & 0.183 & 0.184 & 95.0\% & 92.3\% & 8.41 & 0.130 \\
Low & 0.160 & 0.163 & 94.6\% & 92.4\% & 11.26 & 0.122 \\
Medium & 0.145 & 0.148 & 94.0\% & 93.0\% & 14.03 & 0.117 \\
High & 0.133 & 0.136 & 94.4\% & 93.6\% & 16.96 & 0.112 \\
Very High & 0.125 & 0.129 & 94.4\% & 92.4\% & 19.64 & 0.110 \\
\bottomrule
\end{tabular}
\end{table}

\subsubsection{Variance Component Estimation}

Our estimator demonstrates excellent performance in estimating the $V_E$ component, which captures the measurement error variance from residual outcome noise. Table~\ref{tab:additional-simulation-results} shows the close correspondence between the true $SE_E$ (computed as the standard deviation of $\hat{\tau} - \text{SATT}$ across simulations) and our estimated $SE_E$ values across all overlap scenarios. The differences are minimal, ranging from 0.001 to 0.004, indicating that our pooled variance estimator accurately captures this component of the total variance.

This accuracy is particularly important because the $V_E$ component reflects how matching structure affects variance through control unit reuse. Unlike the bootstrap method, which does not decompose variance into interpretable components, our approach allows researchers to understand how different aspects of matching contribute to overall uncertainty.

\subsubsection{Asymptotic Bias Verification}

The simulation results provide clear evidence of asymptotic bias as predicted by our theoretical propositions. Comparing coverage rates with and without bias correction demonstrates the importance of the bias correction term $B_n$. Across all overlap scenarios, coverage without bias correction is systematically lower than with bias correction:

\begin{itemize}
    \item Very Low overlap: 92.3\% vs 95.0\% (difference of 2.7 percentage points)
    \item Low overlap: 92.4\% vs 94.6\% (difference of 2.2 percentage points)  
    \item Medium overlap: 93.0\% vs 94.0\% (difference of 1.0 percentage points)
    \item High overlap: 93.6\% vs 94.4\% (difference of 0.8 percentage points)
    \item Very High overlap: 92.4\% vs 94.4\% (difference of 2.0 percentage points)
\end{itemize}

This pattern confirms that bias correction is essential for achieving proper coverage, particularly in low-overlap scenarios where matching quality is poorer and bias is more substantial.

\subsubsection{Effective Sample Size Analysis}

The effective sample size of controls (ESS$_C$) shows an intuitive increasing pattern with the degree of overlap, ranging from 8.41 in very low overlap scenarios to 19.64 in very high overlap scenarios. This trend reflects that higher overlap allows for more efficient use of the control sample, as each control unit can contribute meaningfully to multiple matches without dramatically inflating variance through excessive reuse.

The mean $V/n_T$ values (representing the estimated total variance scaled by sample size) show a corresponding decreasing pattern as overlap increases, from 0.130 to 0.110. This demonstrates that better overlap not only improves bias (through closer matches) but also reduces variance (through more efficient control utilization), confirming the bias-variance tradeoff in matching estimators discussed in the theoretical sections.

\subsection{Kang and Schafer Simulation Settings}
\label{subsec:ks-dgp}

\subsubsection{Data-Generating Process}

The Kang and Schafer data generating process is structured as follows:

\begin{enumerate}
\item Generate latent covariates $Z_1, Z_2, Z_3, Z_4 \sim \mathcal{N}(0, I_4)$ where $I_4$ is the $4 \times 4$ identity matrix.

\item Calculate propensity scores:
\begin{equation}
p(Z) = \frac{1}{1 + \exp(Z_1 - 0.5Z_2 + 0.25Z_3 + 0.1Z_4)}
\end{equation}

\item Generate treatment assignment as $T \sim \text{Bernoulli}(p(Z))$.

\item Define the outcome model:
\begin{equation}
f_0(Z) = \frac{210 + 27.4Z_1 + 13.7Z_2 + 13.7Z_3 + 13.7Z_4}{50}
\end{equation}

\item Generate potential outcomes with treatment effect $\tau$:
\begin{align}
Y(0) &= f_0(Z) + \epsilon\\
Y(1) &= f_0(Z) + \tau + \epsilon
\end{align}
where $\epsilon \sim \mathcal{N}(0,1)$.

\item The observed outcome is:
\begin{equation}
Y = Y(0)(1-T) + Y(1)T = f_0(Z) + \tau \cdot T + \epsilon
\end{equation}

\item Transform the latent covariates $Z$ to create the observed covariates $X$:
\begin{align}
X_1 &= \exp(Z_1/2)\\
X_2 &= \frac{Z_2}{1 + \exp(Z_1)} + 10\\
X_3 &= \left(\frac{Z_1 \cdot Z_3}{25} + 0.6\right)^3\\
X_4 &= (Z_2 + Z_4 + 20)^2
\end{align}
\end{enumerate}

In our simulations, we use $n=500$ observations and set the true treatment effect $\tau = 0$.

\end{document}